\documentclass{article}[12pt]
\usepackage[english]{babel}
	
\usepackage{lmodern}

\usepackage[top=2cm, bottom=1.8cm,left=2.5cm,right=2.5cm,head=2cm,headsep=0.5cm]{geometry} 

\usepackage{mdframed} 
\usepackage{amsmath}		
\usepackage{amssymb}		
\usepackage{mathrsfs}
\usepackage{latexsym}	
\usepackage{array}
\usepackage{multirow}
\usepackage{amsthm, array, braket, color}
\usepackage{tikz,pgf}
\usepackage{hyperref}
\usepackage{enumitem}
\usepackage{authblk}

\usepackage[toc,page]{appendix} 
\usepackage{fullpage}

\newcommand{\R}{\mathbb{R}}

\newcommand{\C}{\mathbb{C}}
\newcommand{\N}{\mathbb{N}}

\newcommand{\ap}{\rightarrow}
\newcommand{\seg}{\geqslant}

\newcommand{\re}{\text{Re }}
\newcommand{\im}{\text{Im }}
\newcommand{\ran}{\text{Ran }}

\newcommand{\norm}[1]{{\left\|{#1}\right\|}}
\newcommand{\ent}[1]{{\left[{#1}\right]}}

\newcommand{\scal}[1]{{\left\langle{#1}\right\rangle}}

\newcommand{\cala}{\mathcal{A}}

\newcommand{\ff}{\mathcal{F}}

\newcommand{\hh}{\mathcal{H}}

\newcommand{\mm}{\mathcal{M}}
\newcommand{\nn}{\mathcal{N}}
\newcommand{\call}{\mathcal{L}}
\newcommand{\bb}{\mathcal{B}}
\newcommand{\cals}{\mathcal{S}}
\newcommand{\caly}{\mathcal{Y}}

\newcommand{\tr}[1]{\text{Tr}\left( {#1} \right) }
\newcommand{\tra}{\text{Tr} }

\renewcommand{\Re}{\text{Re }}

\newtheorem{prop}{Proposition}
\newtheorem{defi}{Definition}
\newtheorem{theo}{Theorem}
\newtheorem{lemma}{Lemma}

\newmdenv[linecolor=black
          ,topline=false
          ,bottomline=false
          ,rightline=false
          ,leftline=true
          ,leftmargin=0.1cm
          ,linewidth=0.02cm
          ,skipabove=0cm
          ,innerbottommargin=0.05cm
          ,skipbelow=0.05cm
          ]{subproof}

\title{Repeated interaction processes in the continuous-time limit, applied to quadratic fermionic systems\footnote{Work supported by ANR project ``StoQ" N${}^\circ$ANR-14-CE25-0003}}
\author[1]{Simon Andr\'eys}
\affil[1]{Institut Camille Jordan, U.M.R. 5208, Universit\'e Lyon 1, 21 av Claude Bernard, 69622 Villeurbanne cedex, France}

\begin{document}
\maketitle
\vskip -0.5cm
\centerline{ }
\vskip -1mm
\centerline{ }
\vskip -1mm
\centerline{ }
\vskip -1mm
\centerline{ }

\begin{abstract}
We study a class of Lindblad equation on finite-dimensional fermionic systems. The model is obtained as the continuous-time limit of a repeated interaction process between fermionic systems with quadratic Hamiltonians, a setup already used by Platini and Karevski for the one-dimensional XY model. We prove a necessary and sufficient condition for the convergence to a unique stationary state, which is similar to the Kalman criterion in control theory. Several examples are treated, including a spin chain with interactions at both ends.
\end{abstract}
\tableofcontents
\section{Introduction}

  Lindblad equations are one of the many ways to model the behavior of a quantum system in interaction with another one (that is, an open quantum system \cite{OQS} \cite{RivasHuelga}). They arise in the Markovian approach where the second system, called the bath, is not modified by the interaction. The effect of the bath is then taken into account by adding a term to the Schrödinger equation that governs the first system, making it a Lindblad equation; it is the quantum counterpart of the Fokker Plank equation corresponding to a classical Langevin equation (see \cite{briefHistory} for a history of the Lindblad equation). The evolution of the system is no more unitary, but described by a so-called quantum dynamical semigroup $(\Lambda_t)_{t\in [0, \infty)}$ which is made of completely positive trace-preserving maps. This approach has been notably successful in the realm of quantum optics \cite{Carmichael98}, \cite{BoutenGutaMaasen} and to study continuous-time measurement and decoherence. 

A central question in the study of quantum dynamical semigroups is whether the state of the system converges to a stationary state and whether this state is unique. In the case where the bath is at thermal equilibrium, it is expected that the system converges to a Gibbs state at the same temperature as the bath. This article considers the question of convergence and uniqueness on the special class of quadratic fermionic dynamical semigroups, where the stationary state can often be fully described. 

 Quantum dynamical semigroups can be derived in the context of the weak coupling limit, where the interaction between the bath and the system is vanishingly small (see for example \cite{Davies74} or Derezinski and De Roeck \cite{Derezinski2007}). In the case of a larger interaction, the explicit form of the Lindblad equation often needs to be obtained by phenomenological considerations. In \cite{AttalPautrat} Attal and Pautrat  proposed the setup of the continuous-time limit of repeated interactions, in which the bath is modeled as the tensor product of many subsystems that are interacting one after the other with the system. The Lindblad equations studied in this article are obtained by applying this setup to fermionic systems with quadratic Hamiltonians. The equation obtained is of the following form:
\[
\frac{d}{d_t}\rho(t)=-i\big[\, \sum_{i=1}^{2L} T_{i, j}\gamma_i\gamma_j, \, \rho\,\big]-\frac{1}{2}\sum_{1\leq i,j\leq 2L} \left[\Theta M_B \Theta^*\right]_{i, j}\left(\gamma_j \rho \gamma_i-\frac{1}{2}\{\gamma_i\gamma_j, \rho\}\right)
\]
where the $\gamma_i$ are the Majorana operators on the system, where $T$ is a matrix describing the Hamiltonian of the system, where $\Theta$ is a matrix describing the interaction between the system and the bath, and $M_B$ is the covariance matrix of the state on the bath. This kind of evolution is often called a quasi-free semigroup. An important example is the case of the one-dimensional XY model on the spin chain, which can be mapped to a fermionic system. In general, these equations can be thought as describing non-interacting fermions which may jump between the system and the bath. As such, it may seem simplistic and serves more as a toy model to investigate the behavior of open quantum system than as the study of a physical quantum system. These models have the advantage to be explicitely solvable in many case, and still exhibit non-trivial behavior. As an example, the study of the convergence properties on the XY spin chain by Dharhi \cite{Dharhi2008} inspired the improvement of a criterion of convergence for general quantum dynamical semigroups by Fagnola and Rebolledo \cite{FagnolaRebolledoAlgebraicCondition}.

The class of quadratic Lindblad equations on bosonic and fermionic spaces has long been studied, with important examples such as the damped quantum harmonic oscillator, the Dicke laser or the decay of unstable particles (see \cite{AlickiLendi} or \cite{FagnolaRebolledo2002}). Most of these models are on small bosonic systems. More recently, general quasi-free semigroups on fermionic spaces have been studied by Prosen \cite{Prosen2008} \cite{Prosen2010} who described a way to study the convergence and uniqueness properties; his analysis is restricted to the set of even states (i.e. states which commutes with the parity operator). The repeated interaction model on fermionic spaces was introduced by Platini and Karevski \cite{Platini2008} \cite{KarevskiPlatini2009}, with an emphasis on the XY model; they computed explicitly the stationary state in the isotropic case.   The convergence to a unique stationary state for the XY model has been shown in the aforementioned articles of Prosen but also by Dharhi \cite{Dharhi2008} with different methods. Dharhi used a general criterion for the convergence of quantum dynamical semigroup which was developed by Frigerio and Verri \cite{Frigerio1977} \cite{Frigerio1982}. 

The main result of this article is Theorem \ref{theo:uniqueness}. It is a necessary and sufficient condition for the convergence to a unique stationary state for the class of quadratic Lindblad equations on finite-dimensional fermionic spaces. The condition is phrased in terms of the matrices $T$ and $\Theta$ describing the Hamiltonian and the interaction between the bath and the system, and is similar to the Kalman criterion in control theory: there is convergence and uniqueness if and only if the range of $\Theta$ is cyclic by $T$. The theorem result from the general criterion of Fagnola and Rebolledo \cite{FagnolaRebolledoAlgebraicCondition}, but it is valid even if the semigroup admits no faithful invariant state. This is made possible by Proposition \ref{prop:support}, where we show that if $\cals$ is the support of the stationary state with maximal support, then the fermionic space $\hh$ describing the system can be decomposed as the tensor product of two fermionic spaces $\hh_A$ and $\hh_C$ such that $\cals=\{\ket{\Omega_A}\bra{\Omega_A}\}\otimes \hh_C$ where $\ket{\Omega_A}$ is the null state of $\hh_A$.

The criterion is applied to several examples, some of them exhibiting strange behavior: for some specific baths, the stationary state is independent of the Hamiltonian and is the Gibbs state for the number operator. The case of the gauge-invariant spin chain is solved in a slightly more general case than in Karevski and Platini \cite{KarevskiPlatini2009}, and another strange behavior is observed. 
\newline

The article is organised as follows: in a first part, we introduce the basic definitions and theorems on fermionic systems in finite dimension. We introduce the covariance matrix of a state, which is the fermionic analogue of the covariance matrix of a random vector, and the class of quasi-free states, with similar properties as Gaussian vectors.

	In the second part, we describe the continuous-time limit of repeated interaction model, and apply it to quadratic fermionic systems. The properties of the resulting semigroup are exposed: it is shown that the covariance matrix of the system satisfies its own master equation. The condition for the convergence to a unique state is proved.

	Examples in the gauge-invariant case are studied in the fourth part. We show that thermalisation may occur, but not in every case: notably, if the state of the bath is the Gibbs state $e^{-\beta N_B}/Z_B$ where $N_B$ is the number operator on the bath, then the semigroup admits as a stationary state the Gibbs state $e^{-\beta N_S}/Z_S$, where $N_S$ is the number operator on the system, even if the Hamiltonian of the system is not $N_S$. We then study the Gauge-invariant fermionic chain with interactions at both ends, and show that modifying the coupling constant between the bath and the system may change the stationary state in an unexpected way. 

	The last part is about the important example of the XY model. We describe the Jordan-Wigner transform, which allows to map this model to a quadratic fermionic system, and we prove the convergence thanks to our criterion.
\vspace{0.5cm}

{\bf Acknowledgment :} I thank my advisor St\'ephan Attal for introducing this subject to me and helping me throughout the redaction of this article and my co-advisor Claude-Alain Pillet for interesting remarks and discussions about fermionic systems.

\section{Fermionic systems and quasi-free states}

In this section, we recall the classical properties of fermionic spaces, with an emphasis on quadratic Hamiltonians and quasi-free states. The reader may find more details in Araki \cite{Araki1968} or Derezinski and Gerard \cite{DerezinskiGerard}.

\subsection{Basic notations}

For a vector $u$ in a Hilbert space $\hh$, we may use the bra-ket notation : $\ket{u}$ also design the vector $u$, while $\bra{u}$ design the corresponding linear form. The scalar product of two vectors will be written $\scal{u\,,\,v}$. If a conjugation is chosen on $\hh$, the conjugate of a vector is written $\overline{u}$, its real part is $\Re(u)$ and its imaginary part is $\im(u)$.

For an operator $A$ on $\hh$, we write $\ran(A)$ its image, its adjoint is $A^*$, and if there is a conjugation on $\hh$, the conjugate of $A$ is written $\overline{A}$ and its transpose is $A^T=\overline{A}^*$. An operator is self-adjoint if $A^*=A$, symmetric if $A^T=A$, anti-symmetric if $A^T=-A$.

\subsection{Creation and annihilation operators, field operators}

Let $\hh_0$ be a Hilbert space of finite dimension $L$, describing the state of a single particle. The state of an indefinite number of fermions identical to this particle is described by the fermionic space $\hh=\Gamma(\hh_0)$. It is defined that way: let $A(\hh_0)=\bigoplus_{k=0}^L \hh_0^{\otimes k}$, where $\hh_0^{\otimes 0}=\C$. Define $P_{as}$ as the projection over antisymmetric vectors: for any $u_1, ..., u_k \in \hh_0$, we have
\[
P_{a_s} u_1\otimes u_2\otimes ...\otimes u_k= \frac{1}{k!} \sum _{\sigma \in \mathfrak{S}_K} (-1)^{\varepsilon(\sigma)} u_{\sigma(1)}\otimes u_{\sigma(2)}\otimes...\otimes u_{\sigma(k)}
\]
where $\mathfrak{S}_k$ is the group of permutation of $k$ elements and $\varepsilon(\sigma)$ is the signature of the permutation $\sigma$. Then the fermionic space is $\hh=P_{as}A(\hh_0)$. We define the bilinear form $\wedge$ on $\hh$ the following way: for $a \in P_{as} \hh_0^{\otimes k}$ and $b \in P_{as} \hh_0^{\otimes l}$, take 
\[
{a}\wedge{b}=\sqrt{k+l}\, P_{as} {a}\otimes {b}\,.
\]

Any operator $T$ on $\hh_0$ is lifted to an operator $\Gamma(T)=\bigoplus_{n=0}^L T\otimes T\otimes...\otimes T\big|_{\Gamma(\hh_0)}$ on $\hh$. We can also define $d\Gamma(T)=\bigoplus_{n=0}^L \sum_{i=1}^n Id\otimes...\otimes T\otimes ...\otimes Id\big|_{\Gamma(\hh_0)}$ with the property $\exp\left(d\,\Gamma(T)\right)=\Gamma\left(\exp(T)\right)$.

 Given a pure state ${u} \in \hh_0$, we write $c_{{u}}: \hh \ap \hh$ the operator of annihilation of a particle in the state $\ket{u}$ and its adjoint $c_{{u}}^*$ the creation operator. Recall that  for any ${v}={v_1}\wedge...\wedge{v_k} \in\hh$, we have
 \begin{align}
c_{{u}}^* {v}&={u}\wedge{v} \\
c_{{u}}{v}&=\sum_{l=1}^{k}(1)^{l-1} \scal{u,v_l} {v_1}\wedge...\wedge{v_{l-1}}\wedge{v_{l+1}}\wedge...\wedge{v_k}\,.
\end{align}
 The map ${u} \rightarrow c^*_{{u}}$ is $\C$-linear and the map ${u} \rightarrow c_{{u}}$ is $\C$-antilinear. 
 
   In what follows, we fix an orthonormal basis $\ket{1}, \cdots, \ket{L}$ of $\hh_0$. The annihilation and creation operators related to $\ket{i}$ are written $c_i$ and $c_i^*$ for simplicity. They satisfy the anticommutation relations:
\begin{align}
\set{c_i, c_j}&=\set{c_i^*, c_j^*}=0 \\
\set{c_i^*, c_j}&=\delta_{i,j} I
\end{align}
where $\set{A,B}=AB+BA$. Note that the set of $N_i=c_i^* c_i$ for $i=1, \cdots, L$ form a complete set of commuting orthogonal projections. The operator $N=\sum_i N_i$ is called the number operator, it does not depend on the choice of basis of $\hh_S^0$. 

The operators which commute with $(-1)^N$ are called the even operators, they can be written as sum of products of an even number of operators $c_i, c_i^*$. The ones which anti-commute with $(-1)^N$ are called the odd operators, they are sums of products of an odd number of the $c_i$'s and $c_i^*$'s. Odd operators are always of trace $0$.
\newline

We construct an orthonormal basis of $\hh$ corresponding to the basis of $\hh_0$ the following way: for $u_1, ..., u_L \in \{0,1\}$, let $i_1<...<i_k$ be the indexes of $u_i$ with $u_i=1$ and let 
\[
\ket{u_1, ..., u_k}=\ket{i_1}\wedge ...\wedge\ket{i_L}\,.
\]
Then $\set{\ket{u_1, ..., u_L}|, u \in \set{0,1}^L}$ is an orthonormal basis of $\hh$. The expression of the creation and annihilation operators in this basis is
\begin{align*}
c_i \ket{u_1, ..., u_L}&=\delta_{u_i, 1}\, \Pi_{k=1}^{i-1} (-1)^{u_k} \ket{u_1, ..., u_i-1, ..., u_L} \\
c_i^* \ket{u_1, ..., u_L}&=\delta_{u_i, 0}\, \Pi_{k=1}^{i-1} (-1)^{u_k} \ket{u_1, ..., u_i+1, ..., u_L}\,.
\end{align*} 
\newline
 
In order to consider linear combinations of creation and annihilation operators, we would like to consider the map
\[
\begin{array}{ccc}
\hh_0\oplus \hh_0 &\rightarrow &\bb(\hh)\\
{u}\oplus{v}&\mapsto& c^*_{{u}}+c_{{v}}
\end{array}
\]
However this map is antilinear in the second variable, which is not practical. To overcome this, we consider the Hilbert space $\overline{\hh_0}$ which has the same structure of real vector space as $\hh_0$, but the scalar multiplication by $i$ has been replaced by the scalar multiplication by $-i$. In other words, $\hh_0$ is a vector space endowed with an antilinear isomorphism $s: \hh_0 \rightarrow \overline{\hh_0}$ (many authors simply write $s(x)=\overline{x}$ for $x \in \hh_0$, which is natural but may be confusing). The scalar product on $ \overline{\hh_0}$ is $\scal{v\,,\,w}=\scal{s^{-1} (w)\,,\,s^{-1}(v)}$.

Now, we define the complex phase space as the Hilbert space $\caly=\hh_0\oplus \overline{\hh_0}$ and the field operator $\varphi$ by
\[
\begin{array}{cccc}
\varphi:& \caly & \rightarrow & \bb(\hh) \\
& {u}\oplus {v} & \mapsto & c^*_{{u}}+c_{s^{-1}({v})}\,.
\end{array}
\]
It is a $\C$-linear map. 

The complex phase space $\caly$ is endowed with the anti-linear involution $\xi({u}\oplus s({v}))={v}\oplus s({u})$. The "real space" defined using $\xi$ as a conjugation is $\Re \caly=\set{x \in \caly, \xi(x)=x}=\{{u}\oplus s({u}) | u\in \hh_0\}$. Obviously, 
\[
\varphi(\xi(x))=\big(\varphi(x)\big)^*
\]
and the anticommutation relations write 
\[
\{\varphi(x), \varphi(y)\}=\scal{x\,,\,\xi(y)}I\,.
\]

If we fix a Hilbert basis $\ket{1}, ..., \ket{L}$ of $\hh_0$, it allows to define two interesting basis of $\caly$. First, there is the basis with elements $e_i=\ket{i}\oplus0$ and $e_{i+L}=0\oplus s(\ket{i})$, so that 
\begin{align*}
\varphi(e_i)&=c_i^*& \varphi(e_{i+L})&= c_i\,.
\end{align*}
Secondly, consider another basis of $\caly$, which is contained in $\Re \caly$: take $f_i=\, (e_i+e_{i+L})$,  $f_{i+L}=\,(ie_i-ie_{i+L})$ (it is an orthogonal basis, but not a normal one. We choose not to normalise it to be more in phase with conventions in the litterature). Then the operators $\gamma_i=\varphi(f_i)$ (for $i =1, ..., 2L$) are called Majorana operators (or Clifford operators). They are self-adjoint and satisfy the following anticommutation relations: 
\[
\{ \gamma_i, \gamma_j\}= 2\delta_{i, j}. 
\]
The relations between the $c_i$'s and the $\gamma_i$ are: for $1 \leq i \leq L$, 
\begin{align*}
\gamma_i&=c_i+c_i^* & \gamma_{i+L}&=-i (c_i -c_i^*)\\\
c_i&=\frac{1}{2}(\gamma_i+i\gamma_{i+L}) & c_i^*&=\frac{1}{2}(\gamma_i-i\gamma_{i+L})\,.
\end{align*}

We will call the basis $\{e_1, ..., e_{2L}\}$ the creation/annihilation basis and the basis $\{f_1, ..., f_L\}$ the Majorana basis. For an operator $A$ on $\caly$, we will generally write $A_{c}$ and $A_{f}$ its matrices in the creation/annihilation and the Majorana basis respectively. When a basis is chosen and if it does not cause confusion, we identify the matrix with the operator and simply write it $A$.  It is practical to choose which basis to use depending on the context;  to swap between the basis, note that 
\[
A_{f}=\frac{1}{2}\begin{pmatrix} I & I \\ -iI& iI\end{pmatrix} A_{c} \begin{pmatrix} I & iI \\ 1 & -iI \end{pmatrix}\,.
\]

\subsection{Hamiltonian for non-interacting Fermions}

Let $H_0$ be a Hamiltonian on the one-particle space $\hh_0$, with associated matrix $T^0$ in the basis $\ket{1}, \cdots, \ket{L}$. If we make the physical hypothesis that the fermions are non-interacting, then the corresponding unitary evolution on the fermionic space $\hh$ is $\Gamma\left(e^{-itH_0}\right)$ so the corresponding Hamiltonian is $d\Gamma(H_0)$.  It is easily seen that 
\[
d\Gamma(H_0)= \sum_{i,j}  T^0_{i,j} c_i^* c_j\,.
\]
This is a first motivation to study operators that are quadratic in the creation and annihilation operators. We will also encounter more general quadratic operators, with terms of the form $c_i c_j$ and $c_i^* c_j^*$. They arise for example in the case of the $XY$ model for a chain of spin (see Section \ref{sec:XYchain}). The study of quadratic operators is simplified by the introduction of the column operator, as follows.

\subsection{Column operator}

Let us consider the map  $C^*$ from $  \hh \otimes \hh_0$ to $\hh$ defined by
\[
C^*(\ket{u_1, \cdots, u_L}\otimes {v})=c^*_{{v}}\ket{u_1, \cdots, u_L}\,.
\]
 The matrix $C_b^*$ of $C^*$ in a chosen basis $\ket{1}, ..., \ket{L}$ of $\hh_0$ is a line of operators: 
\[
C_b^*=\begin{pmatrix}
c_1^* & \cdots & c_L^*
\end{pmatrix}\,.
\]
Its adjoint $C: \hh_0\rightarrow \hh\otimes \hh_0$ is a column of operators.
Thus the previous formula about the quantization of one-particle Hamiltonian becomes 
\begin{align}
d\Gamma(H_0)=C^* H_0 C. \label{formula:quadratic_Gamma}
\end{align}

The column operator $C$ is not sufficient to treat general quadratic operators because it does not allow to mix the creation and annihilation operators;  we will need the larger vector of operators $F $ from $\hh\otimes \caly$ to $\hh$ defined by 
\[
F^*(\ket{u_1, \cdots, u_L}\otimes y)= \varphi(y)\ket{u_1\cdots, u_L}\,.
\]
In the creation/annihilation basis, the matrix $F_{c}$ of $F$ is
\begin{align}
F_{c}=
\begin{pmatrix}
c_1 \\
\vdots\\
c_L\\
c_1^*\\
\vdots\\
c_L^*
\end{pmatrix}. \label{formula:F_ca}
\end{align}
 
and in the Majorana basis, its matrix $F_{m}$ is 

\begin{align}
F_c=
\begin{pmatrix}
\gamma_1 \\
\vdots\\
\gamma_{2L }
\end{pmatrix}. \label{formula:F_m}
\end{align}

\subsection{Change of basis and Bogoliubov transforms}

Bogoliubov transforms are a way of swapping between different representations of $\hh$ as the second quantization of a space $\hh_0$. They are used notably to simplify quadratic operators.

\begin{defi}
A unitary transform $U$ of $\caly=\hh_0\oplus \hh_0$ is called a Bogoliubov transform if it commutes with the anti-linear involution $\xi$.
\end{defi}

 Bogoliubov transforms are easily characterized by their matrix: 
\begin{prop}
In the creation/annihilation basis, the operator $U$ is a Bogoliubov transform if and only if $U_{ca}$ it is unitary and the form 
\[
\begin{pmatrix}
\gamma & \mu\\
\overline{\mu} & \overline{\gamma}
\end{pmatrix}\,.
\]
where $\gamma$ and $\mu$ are $L\times L $ matrices. 

In the Majorana basis, the operator $U$ is a Bogoliubov transform if and only if $U_{m}$ it is unitary and real. 
\end{prop}
\begin{proof}
Since the Majorana basis is contained in $\Re \caly$, the involution $\xi$ corresponds to the conjugation, so the matrices that commute with $\xi$ are the real matrices. 

Now, for the creation/annihilation basis, note that if ${u}$ and ${v}$ are vectors of $\hh_0$ of matrices $u_b$ and $v_b$ in the choosen basis of $\hh_0$, then the matrix of ${u}\oplus s({v})$ in the creation/annihilation basis is 
\[
\begin{pmatrix} u_b \\ \overline{v_b} \end{pmatrix}
\]
(indeed ${v}=\sum_i (v_b)_i \ket{i}$ so $s({v})=\sum_i \overline{(v_b)_i} s(\ket{i})$ ). Thus, $\xi$ corresponds to the application 
\[
\xi_{c}\begin{pmatrix} u \\ v \end{pmatrix} =\begin{pmatrix} \overline{v} \\ u \end{pmatrix}\,.
\]
It is easily checked that matrices commuting with the map $\xi_{c}$ are of the announced form. 
\end{proof}

Bogoliubov transforms are interesting because of the following property: 
\begin{prop}
Let $U$ be an operator on $\caly$ and let $\tilde{\varphi}=\varphi U$. Then $U$ is a Bogoliubov transform if and only if $\tilde{\varphi}$ satisfies the same adjoint and anticommutations properties of $\varphi$~: 
\begin{align}
\tilde{\varphi}(\xi(y))&=\big(\tilde{\varphi}(y)\big)^* \\
\{\varphi(x), \varphi(y)\}&=\scal{x, \xi(y)}\,.
\end{align}  
\end{prop}
As a consequence, the elements $\tilde{c_i}, \tilde{c_i}^*$ of the column operator $\tilde{F_{c}}=U_{c}^* F_{c}$ satisfies the anti-commutation properties.

The new operators $\tilde{c_i}$ define another fermionic structure on $\hh$; 
but in general, a Bogoliubov transform does not correspond to a $\C$-linear change of basis on the one-particle space $\hh_0$, but only to a change of basis on $\caly$.

 We will show at the end of the next subsection that a Bogoliubov transforms can always be implemented as the action of a unitary on $\hh$: there exists a unitary $V \in \bb(\hh)$ with $\varphi(Uy)=V \varphi(y) V^*$ for all $y \in \caly$.

\subsection{Quadratic Hamiltonian}\label{subsec:quadratic_hamiltonian}

Quadratic Hamiltonians are an important subclass of Hamiltonians; they appear in the context of the second quantization of a one-particle Hamiltonian as well as in the context of the Jordan-Wigner transform. Every Hamiltonians considered in this article are be quadratic. 
\begin{defi}
A \emph{quadratic Hamiltonian} on $\hh$ is a self-adjoint operator of the form
\[
H=F^* T F=\sum_{1 \leq i, j \leq L} (T_{f})_{i, j}\, \gamma_i \gamma_j
\]
where $T$ is an operator of $\caly$.  In other words, it is a homogeneous polynomial of order two in the creation and annihilation operators.  
\end{defi}

The map $ T \mapsto F^* T F$ is not one-to-one and the condition that $H$ is self-adjoint makes possible to impose that $T$ is self-adjoint. But we can impose even more on $T$.  
\begin{prop}
Up to the addition of a constant to $H$, we can assume that $T$ is self-adjoint and that $\xi T \xi =-T$. in the Majorana basis, it means that $T_{f}$ is the form $i R$ where $R$ is a real anti-symmetric matrix of size $2L \times 2L$. In the creation/annihilation basis, it means that $T_{c}$ is the form
\[
T_{c}=\begin{pmatrix}
A & B\\
-\overline{B} & -\overline{A}
\end{pmatrix}
\]
where $A$ is a self-adjoint $L\times L$ matrix and $B$ is an antisymmetric $L\times L$ matrix (recall that antisymmetric means $B^T=-B$ and not $B^*=-B$). \emph{ We write $QF(L)$ the set of such matrices $T_{c}$}.
\end{prop}

The Hamiltonian corresponding to $T$ is
\begin{align*}
H&=\sum_{1\leq i, j \leq 2L} i R_{i, j} \gamma_i \gamma_j \\
&=\sum_{1\leq i, j \leq L} A_{i, j} c_i^* c_j-\overline{A_{i, j}} c_i c_j^*+B_{i, j} c^*_i c^*_j-\overline{B_{i, j}}c_i c_j\,.
\end{align*}

It will be useful to change from the creation/annihilation basis to the Majorana basis; if $H=F^* T_{c} F$ in the creation/annihilation basis, then in the Majorana basis $T$ writes
\[
T_f=\frac{1}{2} \begin{pmatrix}
1 & 1 \\
-i  & i 
\end{pmatrix}
T_{c}
\begin{pmatrix}
1 & i \\
1 & -i
\end{pmatrix}
=i\begin{pmatrix}
\im A+\im B & \re A+\re B \\
-\re A +\re B & \im A-\im B
\end{pmatrix}\,.
\]

The commutator of quadratic operators and of field operators are simple to describe: for any operator $T: \caly \rightarrow \caly$ and for any $x \in \caly$, we have
\[
[F^*TF, \varphi(x)]=\frac{1}{2}\varphi((T-\xi T \xi)x)\,.
\]
Thus, if $\xi T \xi =-T$, we have simply $[F^*TF, \varphi(x)]=\varphi(Tx)$.
\newline

The next proposition expresses the fact that we can reduce any quadratic Hamiltonian by a Bogoliubov transform.
\begin{prop}\label{prop:reduction}
Let $T_{c}  \in QF(L)$ be the matrix of a quadratic Hamiltonian; then there exists a Bogoliubov transform $U$ such that in the creation/annihilation basis $U T U^*$ is of the form 
\[
\begin{pmatrix}
\Lambda & 0 \\
0 & -\Lambda 
\end{pmatrix}
 \]
 where $\Lambda$ is a diagonal $L \times L$ matrix.  As a consequence, if we let $\tilde{c_i}, \tilde{c_i}^*$ be the coefficients of $\tilde{F}_{c}=U_{c} F_{c}$ and $\lambda_i$ the eigenvalues of $T$, then 
\[
H=F^* T F=\tilde{F}^* UTU^* \tilde{F}=\sum_i \lambda_i \left(\tilde{c_i}^*\tilde{c_i}-\tilde{c_i}\tilde{c_i}^*\right)=2\sum_i \lambda_i \tilde{c_i}^* \tilde{c_i}-\sum_i \lambda_i Id\,.
\]
\end{prop}

\begin{proof}
Let $iR$ be the matrix of $T$ in the Majorana basis, where $R$ is an antisymmetric real matrix. It is a classical fact that antisymmetric real matrix of even size can be block-reduced by a real unitary transform: there exists a real unitary $O$ and a diagonal matrix $\Lambda$ of size $L\times L$ with 
\[
O^* R O=\begin{pmatrix}
0 & \Lambda \\
-\Lambda & 0
\end{pmatrix}. 
\]
Let $U$ be the matrix associated to $O$ in the creation/annihilation basis: 
\[
U=\frac{1}{2}
\begin{pmatrix}
1 & i \\
1 & i
\end{pmatrix} O\begin{pmatrix}
1 & 1 \\
-i  & i 
\end{pmatrix}\,.
\] 
Then $O$ is a real unitary matrix so $U$ is a Bogoliubov transform. Moreover, 
\begin{align*}
U^* T U&=\frac{1}{2}
\begin{pmatrix}
1 & i \\
1 & i
\end{pmatrix} O^* \, iR O
 \begin{pmatrix}
1 & 1 \\
-i  & i 
\end{pmatrix} \\
&= \frac{1}{2}
\begin{pmatrix}
1 & i \\
1 & i
\end{pmatrix}
\begin{pmatrix}
0 & i\Lambda \\
-i\Lambda & 0
\end{pmatrix}
  \begin{pmatrix}
1 & 1 \\
-i  1 & i 
\end{pmatrix}  \\
&=\begin{pmatrix}
\Lambda & 0 \\
0 & -\Lambda
\end{pmatrix}\,.
\end{align*}
Hence $U$ diagonalize $T$ as required. 
\end{proof}

{\bf Remark:}  This proposition shows the advantages of using both the creation/annihilation basis and the Majorana basis: on the one hand, in the Majorana basis matrices are generally of a simpler form and we can use classical theorems to reduce them; on the other hand, the creation and annihilation operators are easier to interpret and manipulate. Indeed, since the $c_i^* c_i$ form a family of mutually commuting projectors, this theorem allows to effectively diagonalize $H$. Moreover it makes easy to compute $\exp( H)$:
\[
e^H=e^{-\sum_i\lambda_i}\Pi_i e^{2\lambda_i c_i^* c_i}=e^{-\sum_i \lambda_i} \Pi_i \left(1-(e^{2\lambda_i}-1)c_i^*c_i\right). 
\]

The following proposition allows to compute the effect of $\exp(H)$ on the creation and annihilation operators.
\begin{prop}\label{prop:commutation_H}
Let $H=F^* T F$ be a quadratic Hamiltonian. Let $\alpha \in \C$, then 
\[
\left(e^{-\alpha H}\otimes Id_{\caly} \right)\, F \,\left(e^{\alpha H}\otimes Id_{\caly} \right)=\left(Id_{\hh}\otimes e^{2\alpha T}\right) \, F\,.
\]
If $\alpha \in i\R$, then $e^{2\alpha T}$ is a Bogoliubov transform. Any Bogoliubov transform can be obtained that way. 
\end{prop}
This proves that any Bogoliubov transform can be implemented by a unitary on $\hh$. 

\begin{proof}
This proof could be written without using the reduction of $T$, but it is more convenient with it. 

Let $U$ be a Bogoliubov transform that reduce $T_{c.a}$ as in proposition \ref{prop:reduction} and $\tilde{c_i}, \tilde{c_i}^*$ the elements of $\tilde{F}_{c}=U_{c}^* F_{c}$ and let $\tilde{N_k}=\tilde{c_k}^* \tilde{c_k}$. Write $\lambda=\sum_i \lambda_i$, then 
\[
H=2\sum_k \lambda_k \tilde{N_k}-\lambda\,.
\]
But for any $i, k$, $\tilde{N_k} \tilde{c_i}=(1-\delta_{i, j}) \tilde{c_i}\tilde{N_k}$ so 
\[
H \tilde{c_i}=\tilde{c_i}\left(2 \sum_{k\neq i} \lambda_k\tilde{N_k}-\lambda \right)
\]
and $\tilde{c_i}\tilde{N_i}=\tilde{c_i}$ so
\[
\tilde{c_i}H=\tilde{c_i}\left(2 \sum_{k\neq i} \lambda_k\tilde{N_k}-\lambda \right)+2\lambda_i \tilde{c_i}\,.
\]
For any $A \in \bb(\hh)$, write $L_A$ and $R_A$ the left and right multiplication by $A$. Then we just proved that for all $i$,
\[
\left((L_H-R_H\right)(\tilde{c_i})=-2 \lambda_i \tilde{c_i}\,.
\]
Likewise, we claim that $\left((L_H-R_H\right)(\tilde{c_i}^*)=-2 \lambda_i \tilde{c_i}$. This means that
\[
\left((L_H-R_H)\otimes Id_{\caly}\right) \tilde{F} =-\left(Id_{\bb(\hh)}\otimes 2 U^* T U \right) \tilde{F}
\]
and since $Id \otimes U^*$ commutes with $\left(L_H-R_H\right)\otimes Id_{\caly}$, it implies that 
\[
\left((L_H-R_H)\otimes Id_{\C^{2L}}\right) F= -\left(Id_{\bb(\hh)}\otimes 2  T \right) F\,.
\]
But $\exp \big(-\alpha (L_H-R_H) \big)=L_{\exp(-\alpha H)} R_{\exp(\alpha H)}$, so exponentiating the previous formula gives the announced result. 

For the last part of the proposition, any Bogoliubov transform is a real unitary in the Majorana basis, hence it is the form $e^{2i R}$ where $R$ is a real antisymmetric matrix, that is, a matrix of $QF(L)$ in the creation/annihilation basis. 
\end{proof}

\subsection{Covariance matrix}

Let $\rho$ be a state on $\hh$ (i.e. a positive operator of trace 1). It is represented by a $2^L \times 2^L$ matrix (or an operator on $\caly$); but many of its interesting properties are described by a smaller matrix, called the covariance matrix: it is an $L \times L$ matrix (or an operator of $\hh_0$) with a definition similar to the one of the covariance matrix in probabilities.

\begin{defi}
The \emph{covariance matrix} $\mm$ of a state $\rho$ is the operator on $\caly$ defined by
\[
\mm=\tra_\hh\,\left(\big(\rho\otimes Id_{\caly}\big)\, F F^* \right)\,.
\]
where $\tra_{\hh}$ is the partial trace on $\hh\otimes \caly$.

In the creation/annihilation basis, the covariance $\mm$ is a $2L\times 2L$ matrix and for $1 \leq i, j\leq L$, we have $\mm_{i, j}=\tr{\rho c_i c_j^*}$,  $\mm_{i,j+L}=\tr{ \rho c_i c_j}$ and so on. 

 The \emph{small covariance matrix} of $\rho$ is the operator $\mm^0$ on $\hh_0$ defined by $\mm^0_{i,j}=\tr{\rho c_i c_j^*}$ in the choosen basis of $\hh_0$, or in operator notation
\[
\mm^0=\tra_\hh\,\left(\big(\rho\otimes Id_{\hh_0}\big) C C^*\right)\,.
\]
\end{defi}
The diagonal terms of $\mm$ are easily interpreted: the quantity $\mm_{i, i}=\tr{\rho c_i c_i^*}=1-\tr{ \rho c_i^* c_i}$ is the probability of absence of a particle in the mode $i$. We choose the convention $c_i c_j^*$ and not $c_i^* c_j$ because $C$ is anti-linear under change of basis.  

Due to the anticommutation relations, the covariance matrix is always of the following form in the creation/annihilation basis:
\[
\frac{1}{2} Id_{2L}+Q
\]
 where $Q \in QF(L)$. In the Majorana basis, it is the form 
\[
\frac{1}{2} Id_{2L}+i R
\]
where $R$ is a real antisymmetric matrix. 

Because of Proposition \ref{prop:commutation_H}, it is easy to compute the evolution of the covariance matrix under the evolution generated by a quadratic Hamiltonian. 
\begin{prop}\label{prop:covariance_evolution}
Let $H=F^*TF$ be a quadratic Hamiltonian; let $\rho_t=e^{-itH} \rho e^{itH}$ and let $\mm_t$ be the covariance matrix of $\rho_t$. Then 
\[
\mm_t=e^{2it T} \mm_0 e^{-2it T}\,.
\]
\end{prop}
\begin{proof}
We have 
\begin{align*}
\mm_t&=\tra_{\hh} \left( e^{itH} \rho\, e^{-itH} F F^* \right) \\
&=\tra_{\hh}\left(\rho \left(e^{-itH} F e^{itH}\right)\,\left(e^{-itH} F e^{itH}\right)^* \right) \\
&=\tra_{\hh} \left(\rho\, e^{2itT} F F^* e^{-2itT} \right) \\
&= e^{2itT}\tra_{\hh}\left(\rho F F^*\right) e^{-2itT}/,.
\end{align*}
\end{proof}

This formula, together with the reduction of Proposition \ref{prop:reduction}, means that up to a Bogoliubov transform we can assume that $\mm$ is diagonal. Since $\mm_{i, i}=\tr{\rho c_i c_i^*}\, \in [0, 1]$, it implies that $0 \leq \mm \leq Id$. Moreover, if $\ker \mm \neq 0$ or $\ker (Id-\mm)\neq 0$ then $\rho$ is not faithful (i.e. its support is not $\hh$).

\subsection{Quasi-free state}
 
 Quasi-free states are an important class of states on fermionic spaces; they are Gibbs states for quadratic Hamiltonians and they behave in a similar way as the Gassian states in classical probabilities. In this section, we introduce quasi-free states and describe some of their fundamental properties.
 
 \begin{defi}
A state $\rho$ on $\hh$ is said to be a \emph{Gaussian state} if it has the form
\[
\rho=\frac{1}{Z} \exp(-\beta H)
\]
where $H=F^* T F$ is a quadratic Hamiltonian. A limit of Gaussian states will  be called a \emph{quasi-free state}. A quasi-free state is  a Gaussian state if and only if it is non-degenerate (i.e. of full support).

A quasi-free state is called a \emph{gauge-invariant quasi-free} state if $H$ is gauge-invariant (i.e. the form $C^* T^0 C$), or equivalently if $\rho$ commutes with the number-particle operator $N=\sum c_i^* c_i$. 
 \end{defi}
 A first property of quasi-free states is the invariance of the set of quasi-free states under the evolution generated by a quadratic Hamiltonian.
\begin{prop}\label{prop:invariance_unitary}
Let $\rho$ be a quasi-free state and $H=F^* T F$ be a quadratic Hamiltonian. Then for all $t \in \R$, $\rho_t=e^{itH} \rho e^{-itH}$ is a quasi-free state. 
\end{prop}
 \begin{proof}
By continuity of $e^{itH}$, it is sufficient to show it for Gaussian states. Let us assume $\rho=e^{-\beta F^* R F}/Z$ for some $R \in QF(L)$. Then 
\begin{align*}
\rho_t&= e^{itH} \rho e^{-itH}\\
&=\frac{1}{Z} \exp\left(-\beta e^{itH} F^* R F e^{-itH} \right) \\
&= \frac{1}{Z} \exp\left(-\beta e^{itH} F^* R F e^{-itH} \right)\\
&=\frac{1}{Z}\exp\left(-\beta \left(e^{-itH}F e^{itH}\right)^* R \left(e^{-itH} F e^{itH}\right)\right) & \text{by Proposition \ref{prop:commutation_H}}\\
&= \frac{1}{Z} \exp\left(-\beta F^*\left(e^{2itT} Re^{-2itT}\right) F \right).
 \end{align*}
Hence it is a Gaussian state, generated by the quadratic Hamiltonian $ F^*\left(e^{2itT} Re^{-2itT}\right) F$.
\end{proof}

A second useful property is that we can express the covariance matrix of a Gaussian state as a function of $T$.

\begin{prop}\label{prop:cov_matrix}
If $\rho$ is a non-degenerate quasi-free state of the form
\[
\rho=\frac{1}{Z}e^{-\beta F^* T F}
\]
 then 
\[
\mm=(I+e^{-2\beta T})^{-1}\,.
\]
 In particular, if $\rho=\frac{1}{Z} \exp(-\beta C^* T^0 C)$ is gauge-invariant then
\[
\mm^0=(I+e^{-\beta T^0})^{-1}\,.
\]
\end{prop}
This formula is easily obtained by reducing $T$ by a Bogoliubov transform and using the fact that $\tr{\Pi_{k=1}^l N_{i_k}}=2^l$ where $i_1, ..., i_l$ are two-by-two distinct indexes. 

Note that some authors prefer to define the small covariance matrix as $\mm^0_{i,j}=\tr{\rho c_i^* c_j}$. In this case the formula is $\mm^0=\left(I+e^{\beta T^0}\right)^{-1}$.
\newline

The most characteristic property of quasi-free states is the Wick formula. It expresses the fact that a quasi-free state is fully described by its covariance matrix. This is the non-commutative analogue of the fact that a Gaussian random vector is characterized by its covariance matrix, which justifies the term "Gaussian state" for "quasi-free state".

\begin{theo}\label{theo:wick}
Let $\rho$ be a density matrix on $\hh$. The following assertions are equivalent:
\begin{enumerate}
\item The state $\rho$ is a quasi-free state.
\item The state $\rho$ verifies Wick's formula: for every $a_1, \cdots, a_n \in \text{Vect }(c_i, c_i^* \, | i=1, \cdots, L)$, 
\begin{align}
\tr{a_1 \cdots a_n\rho} 
& = 
0  \text{   if $n$ is odd} \\
\tr{a_1\cdots a_n \rho}
& =  \sum_{\sigma \in \mathcal{P}_2^n}  \varepsilon (\sigma) \underset{i=1}{\overset{n/2}{\prod}} \tr{a_{\sigma(2i)} a_{\sigma(2i+1)}\rho}
 \text{   if $n$ is even} \\
\end{align}
where $\mathcal{P}_2^n$ is the set of pairings of $\set{1, \cdots, n}$, i.e. the set of permutations $\sigma$ of $\set{1, \cdots, n}$ that satisfies $\sigma(2i) < \sigma(2i+1)$.
\end{enumerate}
\end{theo}

A proof that $1$ implies the formula $2$ can  be found for example in an article from Gaudin \cite{Gaudin1960}. It may be obtained by a repeated use of the commutation relation of Proposition \ref{prop:commutation_H}. The fact that $2$ implies $1$ is often omitted, it can be proved as follows. 
\begin{proof}
Assume that $\rho$ satisfy the Wick's formula $2$. Let $\mm$ be its covariance matrix and assume that $0 < \mm < Id$. Let $T=\ln(\mm^{-1}-I)$ and write 
\[
\rho_\mm=\frac{1}{\tr{e^{F^* T F}}} e^{F^* T F}\,.
\]
Then $\rho$ and $\rho_\mm$ have the same covariance matrix and are both satisfying the Wick's formula. Hence $\tr{\rho a_1...a_n}=\tr{\rho_\mm a_1...a_n}$ for every  $a_1, ..., a_n \in\text{Vect }(c_i, c_i^* \, | i=1, \cdots, L) $, which implies that $\tr{\rho A}=\tr{\rho_\mm A}$ for any observable $A$ and so $\rho=\rho_\mm$. Thus, $\rho$ is quasi-free. 

To prove it when $\mm$ has 0 or 1 as eigenvalue, let $0 < \mm_n < Id$ with $\mm_n \ap \mm$. Then for any $a_1, ..., a_n\in \\text{Vect }(c_i, c_i^* \, | i=1, \cdots, L)$, $\tr{\rho_{\mm_n} a_1... a_n} \rightarrow \tr{\rho a_1... a_n}$ because the left-hand side of the Wick formula is continuous in the covariance matrix, hence $\rho_{\mm_n} \rightarrow \rho$. 
\end{proof}

\subsection{Tensor product of fermionic spaces}\label{subsec:tensorProduct}

Let $\hh_S$ and $\hh_B$ be two fermionic space, with one-particle spaces $\hh_S^0$ and $\hh_B^0$ of dimensions $L$ and $K$. The ``exponential formula'' states that
 \[
\hh_S\otimes \hh_B\simeq \Gamma\left(\hh_S^0 \oplus \hh_B^0\right)\,.
\]
 This formula is more subtle than it seems, because there are many possible ``good'' isomorphisms between $\hh_S\otimes \hh_B$ and $\hh_{SB}=\Gamma\left(\hh_S^0 \oplus \hh_B^0\right)$.

 First, note that $\hh_{SB}$ represents fermions that can be in the modes of $\hh_S^0$ or in the modes of $\hh_B^0$. Hence, there are canonical injections of $\hh_S$ and $\hh_B$ into $\hh_{SB}$. Likewise, $\bb(\hh_S)$ and $\bb(\hh_B)$ are naturally isomorphic to unital $C^*$-subalgebras $\cala_S$ and $\cala_B$ of $\bb(\hh_{SB})$. As a consequence, there exists partial traces $\tra_S:  \bb(\hh_{SB}) \rightarrow \bb(\hh_B)$ and $\tra_B:\bb(\hh_{SB}) \rightarrow \bb(\hh_S)$, which are unambiguously defined. However, $\cala_S$ and $\cala_B$ are not commuting subalgebra of $\bb(\hh_{SB})$, so there exists no isomorphisms of $C^*$-algebras $E:\bb(\hh_{SB})\rightarrow \bb(\hh_S)\otimes \bb(\hh_B)$ with the property that $E(\cala_S)=\bb(\hh_S)\otimes \set{Id_B}$ and $E(\cala_B)=\set{Id_S}\otimes \bb(\hh_B)$.

 Instead, we have to make a choice. The two most natural isomorphisms are $E_{SB}$ and $E_{BS}$, defined as follows: write $c_i^S, i=1...L$ the creation operators of $\hh_S$, $c_i^B, i=1...K$ the creation operators of $\hh_B$, $N_S=\sum_i (c_i^S)^*c_i^S$ the number operator of $\hh_S$ and $N_B=\sum_i (c_i^B)^* c_i^B$ the number operator of $\hh_B$. Then $E_{SB}$ and $E_{BS}$ are the $C^*$-algebra morphisms defined by:
 \begin{align*}
 E_{SB}\left(c_i^S\right)&=c_i^S\otimes Id_B  &  E_{BS}\left(c_i^S\right)&=c_i^S\otimes (-1)^{N_B} \\
 E_{SB}\left(c_i^B\right)&=(-1)^{N_S}\otimes c_i^B  &  E_{BS}\left(c_i^B\right)&=Id_S\otimes c_i^B.
 \end{align*}

The covariance matrix and quasi-free states behave well with respect to these decompositions, as shown in the following proposition:
\begin{prop}
Let $E$ be the isomorphism $E_{SB}$ or $E_{BS}$ and let $\rho_{SB}$ be a state on $\hh_{SB}$. Its covariance matrix $\mm_{SB}$ is an operator on the phase space $\caly_{SB}=\caly_{S}\oplus\caly_{B}$. 

Then the covariance matrix $\mm_S$ of $\rho_S=\tra_{\hh_B}\left(E(\rho_{SB})\right)$ is the restriction of $\mm_{SB}$ to $\caly_{S}$. Moreover, if $\rho_{SB}$ is quasi-free, then $\rho_S$ is also quasi-free

 Now, if $\rho_S$ and $\rho_B$ are two states on $\hh_S$ and $\hh_B$, then $\rho_{SB}=E^{-1}(\rho_S\otimes \rho_B)$ has $\mm_S\oplus \mm_B$ as a covariance matrix, and if $\rho_S$ and $\rho_B$ are quasi-free then $\rho_{SB}$ is a quasi-free state on $\hh_{SB}$.
\end{prop}

The difference between $E_{SB}$ and $E_{BS}$ is often of no consequence; for example, if $C$ and $D$ are even operators, then $C$ and $D$ commute as elements of $\bb(\hh_{SB})$ and $E_{SB}(CD)=E_{BS}(DC)=C\otimes D$. In particular, quasi-free states and the density matrix behaves well under tensor products and partial traces, whatever is the isomorphism we choose: for $\rho_S$ and $\rho_B$ quasi-free states, $E_{SB}^{-1}(\rho_S\otimes \rho_B)=E_{BS}^{-1}(\rho_S\otimes \rho_B)=\rho_S \rho_B$ is a quasi-free state. Moreover, if $\rho_{SB}$ is a quasi-free state on $\hh_{SB}$, then $\tra_B(\rho_{SB})$ is a quasi free state. The covariance matrix also behave in a way that is independent of the chosen isomorphism: for any states $\rho_{S}$ and $\rho_B$, the covariance matrix of $E_{SB}^{-1}\left(\rho_S\otimes \rho_B\right)$ is the same as the one of $E_{BS}^{-1}\left(\rho_S\otimes \rho_B\right)$ and is simply the direct sum of the covariance matrices of $\rho_S$ and $\rho_B$.  
\newline

{\bf Warning on more exotic isomorphisms: } The isomorphisms $E_{SB}$ and $E_{BS}$ are not the only possibilities. For example, if we divide $\hh_B^0$ as the direct sum of two subspaces, $\hh_B^0=\hh_{B_L}^0\oplus \hh_{B_R}^0$, we could consider $E_{B_L S B_R}: \bb(\hh_{SB})\rightarrow \bb(\hh_{S})\otimes \bb(\hh_{B})$ defined by $E(B_L A B_R)=A\otimes B_LB_R$ for any even operators $B_L \in \bb(\hh_{B1})$, $B_R \in \bb(\hh_{B2}), A \in \bb(\hh_S)$. This type of isomorphism will be used in the context of spin chains (see Section \ref{sec:XYchain}). Note that if $\rho_B$ is bipartite, $\rho_B=\rho_{B1}\otimes \rho_{B2}$ with $\rho_{B_1}, \rho_{B2}$ quasi-free state, then if $\rho_S$ is a quasi-free state of $\hh_S$, $E_{B_1 S B_2}^{-1}\left(\rho_{B}\otimes \rho_S\right)$ is a quasi free state; but if $\rho_B$ is a non-bipartite quasi-free state, $E_{B_1 S B_2}^{-1}\left(\rho_{B}\otimes \rho_S\right)$ might not be quasi-free.

\section{Repeated interactions processes and fermions}\label{sec:repeated_interaction}

In this section, we describe a model for a continuous-time dissipative dynamical evolution on a fermionic system. First, we introduce the general repeated interaction model, which we then apply to fermionic systems. Some important properties are derived, namely the fact that quasi-free states remain quasi-free under the evolution and that the covariance matrix follows a closed master equation. Finally, we derive a necessary and sufficient criterion for the convergence to a unique stationary state. Examples and applications are described in the following section. 

\subsection{Quantum dynamical semigroup}

Quantum dynamical semigroups (or Lindblad semigroups) are a general model for quantum systems in interaction with a large bath. It is the quantum counterpart of a classical Markov semigroup, and it corresponds to the situation where the bath is left unperturbed by the system: its effect on the system is the same at all times. 
\begin{defi}
Let $\hh_S$ be a finite-dimensional Hilbert space. A quantum dynamical semigroup on $\hh_S$ is a semigroup $(\Lambda^t)_{t\in \R_+}$ where for all $t$, the super-operator $\Lambda^t$ is a quantum channel (i.e. a completely positive, trace preserving map on $\bb(\hh_S)$) and $t\rightarrow \Lambda^t$ is continuous,  $\Lambda_0=Id_{\bb(\hh_S)}$ and $\Lambda^{t+s}=\Lambda^t \Lambda^s$ for all $t, s \geq 0$. 
\end{defi}
The generators of quantum dynamical semigroups are called Lindbladians. Gorini, Kossakowski, Sudarshan and Lindblad established that there exists a self-adjoint operator $H_S \in \bb(\hh_S)$ and a completely positive map $\Phi$ on $\bb(\hh_S)$ so that 
\[
\frac{d}{dt} \Lambda^t \rho= -i[H_S, \rho]+\Phi^*(\rho)-\frac{1}{2} \set{\Phi(Id_S), \rho}. 
\]

There are several ways of obtaining a quantum dynamical semigroup from an explicit model of interaction with a quantum bath. A classical one is Van Hove's weak coupling limit (see \cite{Davies74} and \cite{Derezinski2007}). The one we consider is the continuous-time limit of a repeated interaction process, a model introduced by Attal and Pautrat \cite{AttalPautrat}.

\subsection{Continuous limit of a repeated interaction process}\label{subsec:continuous_limit}

A repeated interaction process describes a system in interaction with a large bath the following way:

\begin{enumerate}
\item The system is described by a Hilbert space $\hh_S$ and a Hamiltonian $H_S$.

\item The bath is composed of an infinite number of copies of the same model $\hh_B$, with Hamiltonian $H_B$  and initial density matrix $\omega$: the total bath space is $\hh_{B, tot}=\bigotimes_{n \in \N^*} \,\hh_n$ where $\hh_n=\hh_B$. The initial state of the bath is $\omega_{tot}=\omega \otimes \omega \otimes \cdots$ and the Hamiltonian of the bath is $H_{B,tot}=\sum H_n$ where $H_1=H_B \otimes Id\otimes \cdots$, $H_2=Id \otimes H_B \otimes Id\otimes \cdots$, $H_3=Id\otimes Id \otimes H_B \otimes Id \otimes \cdots$ and so on. 

\item The system $\hh_S$ interact repeatedly with the subsystems $\hh_n$ of the bath, during the same time length $\tau$ and following the same interaction $\lambda V \in \bb(\hh_S\otimes \hh_B)$. Hence, the global Hamiltonian is constant in the time intervals $[n \tau, (n+1)\tau ]$ and is equal to $H_S+\lambda V_n+ H_n$. 

\item We also assume that $\tra_B (V I \otimes \omega)=0$. This is equivalent to the fact that for all $\rho$, $\tra_B(V \rho \otimes \omega)=0$. This hypothesis will be needed to ensure the convergence when the time scale $\tau$ goes to $0$.  \label{assumption:V_interaction}
\end{enumerate}

This picture means that a quantum channel is repeatedly applied onto the system. Indeed, define the quantum channel $L_\tau$ by
\[
L_\tau(\rho_S)=\tra_{\hh_B} \left( e^{-i\tau(H_S+\lambda V + H_B)} \, (\rho_S \otimes \omega) \, e^{i\tau (H_S+ \lambda V + H_B)} \right)\,.
\]
Then the state of the system at time $n\tau$ is 
\[
\Lambda_\tau (n\tau) (\rho_S)=L_\tau^n(\rho_S)\,.
\]
For general $t \geq 0$, we put $\Lambda_\tau (t)=L_\tau^{\ent{t/\tau}}$ for any $t \geq 0$. 

We want to take the continuous limit of such a process when $\tau \ap 0$ under suitable renormalization.

\begin{prop}\label{prop:limite_continue}
For every $\tau>0$, $t\geq 0$, consider the completely positive map $\Lambda_\tau(t)$ as above. Suppose $H_S$, $H_B$ and $V$ does not depend on $\tau$ and $\lambda=1/\sqrt{\tau}$. Then  the following limit exists for all $t\geq 0$:
\[
\Lambda^t=\lim_{\tau \ap 0} \Lambda_\tau(t)\,.
\]
 The family $\big(\Lambda^t\big)_{t\in [0, \infty[}$ form a quantum dynamical semigroup, with generator:
 \[
 \call(\rho)=-i[H_S, \rho]+\Phi^*(\rho)-\frac{1}{2}\set{\Phi(I), \rho}
 \]
 where $\Phi^*(\rho)=\tra_B\left(V(\rho\otimes \omega) V\right)$. 
\end{prop}

 This result is analogous to Theorem 21 in Attal and Pautrat \cite{AttalPautrat}.
 \begin{proof}
 Write the second order expansion of $e^{i(\tau H_S+  \sqrt{\tau} V + \tau H_B)}$ as $\tau \rightarrow 0$. We obtain
 \[
 L_\tau(\rho)=\rho-i\sqrt{\tau} [\tra_B (V I \otimes \omega), \rho]+\tau \call(\rho)+o(\tau)
 \]
 where $\call$ is defined in the proposition. Now the term of order $\sqrt{\tau}$ is zero because of the assumption \ref{assumption:V_interaction}, thus
 \[
 \Lambda_\tau(t)=\left(Id+\tau \call\right)^{\ent{t/\tau}}+o(\tau)
 \]
 which converges towards $e^{t\call}$ as $\tau \rightarrow 0$.  
 \end{proof}

 Now that we described this general way to construct a quantum dynamical semigroup from a repeated interaction model, we will apply this to the case of fermionic systems whose evolution is generated by quadratic Hamiltonians.

\subsection{Fermionic repeated interaction process} \label{subsec:fermionic_repeated_process}

From now on, the spaces $\hh_S$ and $\hh_B$ are fermionic systems of respective lengths $L$ and $K$, with field operators $\gamma_{1}^S, \cdots, \gamma_{L}^S$ for $\hh_S$ and $\gamma^B_1, \cdots, \gamma^B_K$ for $\hh_B$. We write
$F_S$, $F_B$ and $F_{SB}$ the column operator for $\hh_S$, $\hh_B$ and $\hh_{SB}=\hh_S\otimes \hh_B$ respectively.

 We use one of the isomorphisms $E_{SB}$ and $E_{BS}$ to identify $\hh_{S}\otimes \hh_B$ with $\Gamma(\hh_S^0 \oplus \hh_B^0)$. 
\newline

 {\bf In what follows, we will use mainly the Majorana basis and not the creation/annihilation basis, and we identify an operator on $\caly$ and its matrix in the Majorana basis.  The creation/annihilation basis will only be used in the (important) case of Gauge-Invariant Hamiltonians. }
\newline
 
 To separate the bath and the system in the notation, we will order the elements of $F_{SB}$ this way: 
 \[
 F_{SB}=\left( \gamma_1^S, ..., \gamma_{2L}^S, \gamma_1^B, ..., \gamma_{2K}^B\right)^T.
 \]

 The Hamiltonians used in the repeated interaction process are quadratic Hamiltonians: 
 \begin{align}
 H_S&=\frac{1}{2}F_S^*\,T_S \,F_S \\
 H_B&=\frac{1}{2}F_B^*\, T_B \,F_B\\
 V&=F_S^* \Theta F_B=\frac{1}{2} \left(F_S^* \Theta F_B+F_B^* \Theta^* F_S\right)
 \end{align}
 where, in the Majorana basis, $T_S=iR_S$ and $T_B=iR_B$ with $R_S$, $R_B$ real antisymmetric matrices of size $2L\times 2L$ and $2K \times 2K$ and $\Theta=iW$ where $W$ is a real matrix of size $2L\times 2K$.  We also assume that the state $\omega$ of $\hh_B$ is a quasi-free state and we write $M_B$ its (full) covariance matrix. 
\newline
 
 We claim that $\tra_B(V I \otimes \omega)=0$ (Assumption \ref{assumption:V_interaction} of Subsection \ref{subsec:continuous_limit}). Indeed, since $\omega$ is gaussian, it is an even operator. Moreover, $V$ is a sum of terms the form $\gamma^S_i \gamma^B_j$ or $\gamma^B_i \gamma^S_j$. But 
 \begin{align*}
 \tra_B(\gamma^S_i \gamma^B_j I \otimes \omega)&=\gamma^S_i \tr{\gamma^B_j \omega}
 \end{align*}
and $\gamma^B_j \omega$ is an odd operator, since $\omega$ is even. Thus, it is of trace $0$, so  $\tra_B(\gamma^S_i \gamma^B_j I \otimes \omega)=0$.
 \newline
 
 The total Hamiltonian during one interaction is
 \[
 H_{tot}=\frac{1}{2} F_{BS}^* 
 \begin{pmatrix}
 T_S & \lambda \Theta \\
 \lambda \Theta^*  & T_B
 \end{pmatrix}
 F_{BS}.
 \]
 
With $\lambda=1/\sqrt{\tau}$, it is ensured that the repeated interaction process converges to a quantum dynamical semigroup as $\tau \ap 0$. Let $\big(\Lambda^t\big)_{t\geq 0}$ be the completely positive semigroup describing this evolution.

\subsection{The Lindblad operator for a general state}
By Proposition \ref{prop:limite_continue}, we know that the generator $\call$ of  $(\Lambda^t)_{t\geq 0}$ is the form 
\[
 \call(\rho)=-i[H_S, \rho]+\Phi^*(\rho)-\frac{1}{2}\set{\Phi(I), \rho}
 \]
 where $\Phi^*(\rho)=\tra_B\left(V\rho(\otimes \omega) V\right)$. We will show in the next subsection that knowing the exact form of $\call$ is actually not necessary when we are only interested in the covariance matrix of $\rho$; still, for it will be useful to write $\Phi^*$ as a function of the $c_i, c_i^*$ for the study of states that are not quasi-free. The expression of $\Phi^*$ depends on which of the isomorphisms $E_{SB}$ and $E_{BS}$ we used to identify $\hh_S\otimes \hh_B$ and $\Gamma(\hh_S^0\oplus \hh_B^0)$.  
 
 \begin{prop}\label{prop:evolutionRho}
 Assume we used the isomorphism $E_{SB}$. Then for any state $\rho$, 
 \[
 \Phi^*(\rho)=\sum_{1\leq i, j \leq 2L}  \left(\Theta M_B \Theta^*\right)_{i, j} \gamma_j (-1)^{N_S} \rho (-1)^{N_S} \gamma_i. \label{eq:rho}
 \]
 If we used the isomorphism $E_{BS}$ instead, we obtain 
 \[
  \Phi^*(\rho)=\sum_{1\leq i, j \leq 2L}  \left(\Theta M_B \Theta^*\right)_{i, j} \gamma_j \rho\, \gamma_i.
 \]
 \end{prop}
 
 Note that if $\rho$ is even, it commutes with $(-1)^{N_S}$ and the choice of isomorphism does not make any difference.
 
 \begin{proof}
Let us show it for the $E_{SB}$ isomorphism.

Replace $V=F_S^* \Theta F_B$ in $\Phi^*(\rho)=\tra_B(V\, (\rho\otimes \omega) \,V)$. We have $V=\sum_{i, j} \Theta_{i, j} \gamma_i^S \gamma_j^B=\sum_{i, j} \Theta_{i, j} \gamma_i^S (-1)^{N_S} \otimes \gamma_j^B$ so 
\begin{align*}
\Phi^*(\rho)&=\sum_{i, j, k, l}\Theta_{i, j}\Theta_{k,l} \gamma_i^S(-1)^{N_S} \rho \gamma_k^S (-1)^{N_S} \, \tr{\gamma_j^B \omega \gamma_l^B} \\
&=\sum_{i, j, k, l} \gamma_i^S(-1)^{N_S} \rho (-1)^{N_S} \gamma_k^S \,\,\Theta_{k, l}\,(M_B)_{l, j}\,\Theta^*_{j, i} & \text{because $\Theta^*=-\Theta^T$ and $ (-1)^{N_S} \gamma_k^S=- \gamma_k^S(-1)^{N_S}$} \\
&= \sum_{i, k} \gamma_i^S(-1)^{N_S} \rho (-1)^{N_S} \gamma_k^S \,\, \left(\Theta M_B \Theta^*\right)_{k, i}
\end{align*}
which is the announced formula. 

The proof in the case of the $E_{BS}$ isomorphism is similar.
 \end{proof}
 
 Since $\Theta M_B \Theta^*$ can be diagonalized by a Bogoliubov transform, this formula allows to write $\Phi^*(\rho)$ in the form $\sum L_i \rho L_i^*$, which will be useful to study the convergence property of $\Lambda^t$.

\subsection{Evolution of the covariance matrix} \label{subsec:evolution_covariance_matrix}
Given an initial state $\rho(0)$, the state of the system at the time $t$ is $\rho(t)=\Lambda^t \rho(0)$. Let us study the evolution of the covariance matrix $\mm(t)$. We will show that $\mm(t)$ follows its own master equation. 

\begin{theo}
Let $\mm(t)$ be the covariance matrix of $\rho(t)$ and $\mm_B$ the covariance matrix of $\omega$, the state of each subsystem of the bath. Then
\begin{align} \label{eq:mm}
\frac{\partial \mm(t)}{\partial t} = -i [T_S, \mm(t) ]- \frac{1}{2} \set{\Theta \Theta^*, \mm(t) } +\Theta \mm_B \Theta^*\,.
\end{align}
\end{theo}

\begin{proof}
It is sufficient to prove Formula \ref{eq:mm} at $t=0$ since the evolution of $\rho$ is given by a semi-group. It can be proved by using Formula \ref{eq:rho} and the commutation relations, but we choose a more direct way. 

We come back to the description of the evolution as the limit of a repeated interaction process. Let $\mm_{SB}(\tau)$ be the covariance matrix of 
\[
\rho_SB(\tau)=e^{i\tau H_{tot}}\, (\rho_S(0) \otimes \omega) \, e^{i\tau H_{tot}}\,.
\]

Then $L_\tau(\rho_S)=\tra_{\hh_B}(\rho_SB(\tau))$, so the covariance matrix of $L_\tau(\rho_S)$, written $\mm_\tau(\tau)$ is simply the first $2L \times 2L$-block of $\mm_{SB}(\tau)$. Let 
\[
R(t)=e^{-i\tau \begin{pmatrix}
 T_S & \frac{1}{\sqrt{\tau}} \Theta \\
 \frac{1}{\sqrt{\tau}} \Theta^*  & T_B
 \end{pmatrix} }\,.
\]
By Proposition \ref{prop:covariance_evolution}, 
\begin{align*}
\mm_{SB}(\tau)&=R(\tau) \mm_{SB}(0) R^* (\tau) \\
&=R(\tau) \begin{pmatrix}
\mm(0) & 0 \\
0 & \mm_B
\end{pmatrix} R^*(\tau)\,.
\end{align*}
But
\[
R(\tau)=I-i\sqrt{\tau} 
\begin{pmatrix}
0 & \Theta\\
\Theta^* & 0
\end{pmatrix}
-i \tau \begin{pmatrix}
T_S & 0 \\
0 & T_B
\end{pmatrix}
-\frac{\tau}{2}
\begin{pmatrix}
\Theta \Theta^* & 0 \\
0 &\Theta^* \Theta
\end{pmatrix}
+ o(\tau)
\]
so the first block of $\mm_{SB}(\tau)$ is
\[
\mm_\tau(\tau)=
 \mm+\tau \left(   -i [T_S, \mm(t) ]- \frac{1}{2} \set{\Theta \Theta^*, \mm_S(t) } + \Theta \mm_B \Theta^*  \right) + o(\tau)\,.
\]
The proposition is obtained by taking the limit $\tau \ap 0$.
\end{proof}

 \subsection{Evolution of a quasi-free state}\label{subsec:evolution_gaussian_matrix}
We are now able to describe the evolution of the covariance matrix, but it doesn't fully describe the state of the system, unless if $\rho_S(t)$ is quasi-free. Luckily, quasi-free states remain quasi-free under the evolution:  
\begin{prop}
If $\rho_S(0)$ is a quasi-free states then $\rho_S(t)$ is a quasi-free states at all $t\seg 0$. 
If moreover the state of the bath $\omega$ and $\rho_S(0)$ are nondegenerate, then $\rho_S(t)$ is nondegenerate for all $t>0$. 
\end{prop}

\begin{proof}
Again, we come back to the description as the limit of a repeated interaction process.
Let $\rho$ be a quasi-free state. Then $\rho\otimes \omega$ is also a limit of quasi-free states and by Proposition \ref{prop:invariance_unitary}, the matrix
\[
\rho_{SB}(\tau)=e^{-i\tau H_{tot}}\, \rho(0) \otimes \omega \, e^{i\tau H_{tot}}
\]
is also a quasi-free state. Thus its partial trace $L_\tau(\rho(0))=\tra_{hh_B}(\rho_{SB}(\tau)$ is a quasi-free state. Hence $L_\tau^n(\rho)$ is a quasi-free states for all $n \in \N$. Wick's formula is also preserved by the limit $\tau \ap 0$, so in the continuous-time evolution, $\rho(t)$ is a quasi-free state.

Now, if $\rho_S(0)$ and $\omega$ are nondegenerate, then their covariance matrix satisfy $\varepsilon Id \leq \mm \leq (1-\varepsilon) Id$ for some $\varepsilon >0$. Then by construction $\varepsilon Id \leq\mm_\tau(n\tau) \leq (1-\varepsilon) Id$ for all $n \in \N$. So $\varepsilon Id \leq\mm(t) \leq (1-\varepsilon) Id$ for all $t$, so $\rho_S(t)$ is nondegenerate.
\end{proof}

\subsection{On the support of a stationary state}

A stationary state is a state $\sigma$ which is invariant under $\lambda^t$.  Semigroups that admits a stationary state with full support have several nice properties, and some authors concentrate on this case. It is motivated by the following facts: 
\begin{itemize}
\item For any CPTP semigroup, the support of a stationary state is stable by the semigroup
\item There exists a stationary state $\sigma$ with maximal support (i.e. for any stationary state $\rho$, we have $\text{Supp } \rho \subset \text{Supp } \sigma$).
\end{itemize} 
Thus we can hope to reduce the study of $(\Lambda^t)_t$ to the study of its restriction to the support of a stationary state. The following proposition shows that the restricted semigroup can still be seen as a semigroup on a fermionic space.

\begin{prop}\label{prop:support}
Let $\sigma$ be a stationary state for $(\Lambda^t)_t$ with maximal support $\cals \subset \hh$. Then there exists a Bogoliubov transform $U$ implemented by a unitary $V$ on $\hh$ and a decomposition of $\hh_0$ as $\hh_{A}^0\oplus \hh_{C}^0$ with the following property: identify $\hh$ with $\Gamma(\hh_A^0)\otimes \Gamma(\hh_C^0$, and let $\ket{\Omega_A}$ be the empty state on $\Gamma(\hh_A^0)$. Then 
\[
V\cals=\set{\ket{\Omega_A}\bra{\Omega_A}}\otimes \Gamma(\hh_C^0)\,.
\]
Moreover, let $(\Lambda_C^t)_t$ be the restriction to $\Gamma(\hh_C^0)$ of the semigroup $\rho\mapsto V \Lambda_t(V^*\rho V)V^*$. Then $(\Lambda_C^t)_t$ can be obtained as the continuous time limit of a repeated interaction process on quasi-free fermionic systems.
\end{prop}

\begin{proof}
Let $\mm$ be the covariance matrix of $\sigma$. There exists a Bogoliubov transform such that in the creation/annihilation basis, $\mm$ is diagonal, and we can also single out its eigenspaces for $0$ and $1$: let $a=\dim \ker (\mm)$ and $c=L-a$, then up to a Boboliubov transform $U$, we can assume that in the creation/annihilation basis we have
 \[
 \mm =\begin{pmatrix}
  Id_a & 0&0 &0\\
 0&  \Lambda_C&0&0\\
 0&0& 0 & 0 \\
0& 0 & 0 & Id_c-\Lambda_C
 \end{pmatrix}
 \]
where $\Lambda_C$ is a $c\times c$ covariance matrix with $0<\Lambda_C<Id$. We can decompose $\hh_S=\Gamma(\C^L)$ as $\Gamma(\C^a)\otimes \Gamma(\C^c)$. Let $\ket{\Omega_A}$ be the empty state of $\Gamma(\C^a)$ and let $\sigma_A=\tra_{\,\Gamma(\C^c)} \,\sigma$. The small covariance matrix of $\sigma_A$ is $Id_a$, hence $\sigma_A$ is the pure state $\ket{\Omega_A}\bra{\Omega_A}$. This implies that $\sigma$ is the  form $ \ket{\Omega_A }\bra{\Omega_A}\otimes \sigma_C$. 
 \begin{subproof}
 Indeed, let $P=\ket{\Omega_A} \bra{\Omega_A}\otimes Id_{\hh_C}$ be a projection of $\hh_S$. Then $\tr{\sigma P}=\tr{\bra{\Omega_A}\sigma_A \ket{\Omega_A}}=1$, so $\tr{\sigma-P\sigma P}=0$, but $P\sigma P \leq \sigma$ so $\sigma=P\sigma P= \ket{\Omega}\bra{\Omega}\otimes \sigma_C$.  
 \end{subproof} 
Hence, $Supp( \sigma) \subset  \set{\ket{\Omega_A}\bra{\Omega_A}}\otimes \Gamma(\hh_C^0)$.

 To prove the other inclusion, we use the fact that $\sigma$ is a stationary state with maximal support.  Note that since $\sigma$ is stationary, its covariance matrix satisfies
\[
-i [T_S, \mm]- \frac{1}{2} \set{\Theta \Theta^*, \mm} +\Theta \mm_B \Theta^*=0
\]
which implies that the quasi-free state $\rho$ of covariance matrix $\mm$ is also stationary and of support $\set{\ket{\Omega_A}\bra{\Omega_A}}\otimes \Gamma(\hh_C^0)$. Since the support of $\sigma$ is maximal this implies that $V \subset Supp(\sigma)$. 
\newline

Now, consider the bath space $\hh_B$, the Hamiltonian $H_S$ and the interaction $V$ used in the construction of $(\Lambda^t)_t$. Then $(\Lambda_C^t)_t$ arise from the continuous time limit of repeated interactions with $\hh_{B'}=\Gamma(\C^a)\otimes \hh_B$ with Hamiltonian $\tra_{\Gamma(\C^c)}(H_S)$ and interaction $V$. 
\end{proof}

We are now ready to study the ergodic properties of $(\Lambda^t)_t$.

\subsection{Convergence toward a unique stationary state}\label{subsec:unique_stationary_state}
We prove a necessary and sufficient condition for what is sometimes called the "return to equilibrium" property: there exists a unique stationary state $\rho_\infty$ and for any initial state $\rho$, $\rho(t)$ converges towards $\rho_\infty$. This property implies notably that the stationary state is quasi-free (since the set of quasi-free states is preserved by the evolution).

\begin{theo}\label{theo:uniqueness}
Consider the continuous process described above. Then the following assertions are equivalent:
\begin{enumerate}[label=(\arabic*)]
\item There exist a unique stationary state $\rho_\infty$; it is quasi-free and $\rho(t) \ap \rho_\infty$ for any initial condition $\rho(0)$. \label{as:uniq}
\item $\ker \Theta^*$ contains no eigenvectors of $T_S$. \label{crit:spectre}
\item The space $V(T_S, \Theta)=Vect\left( \cup_{k=0}^{2L}(\,(T_S)^k\, \Theta\,\right)$ is the whole $\C^{2L}$.\label{crit:kalman}
\end{enumerate}

\end{theo}

The condition $V(T_S, \Theta)=\C^{2L}$ is called the Kalman condition in control theory. It is the condition that we will use in most concrete examples. 

  Prozen proved a similar convergence theorem for fermionic Lindblad evolutions in \cite{Prosen2008}, but it only works with initial conditions that are even operators. His methods are different than ours: Prosen's trick is to see $\bb(\hh_S)$ as $\Gamma(\hh_S^0 \oplus \hh_S^0)$. 
  
  In the following proof, the convergence for quasi-free initial state comes quite easily, but to deal with general states, we resort to the ergodic theory of quantum semigroup, with a few technicalities. 
  \newline

\begin{proof}

We break the proof in several lemmas. 

First, an easy lemma of linear algebra: 

\begin{lemma}
The assertions \ref{crit:spectre} and \ref{crit:kalman} in the theorem are equivalent. 
\end{lemma}

\textit{Proof of the lemma.}
\begin{subproof}
 Note that $V(T_S, \Theta)$ is the smallest subspace of $\C^{2L}$ which is stable by $T_S$ and contains $\ran \Theta=\left(\ker \Theta^*\right)^\perp$. Hence $V(T_S, \Theta)^\perp$ is the maximal subspace contained in $\ker \Theta^*$ and stable by $T_S$, so it is nonzero if and only if $\ker \Theta^*$ contains an eigenvalue of $T_S$.
\end{subproof}

Now, let us study the convergence of covariance matrices.
\begin{lemma}\label{lem:critMatrix}
The assertion \ref{crit:spectre} is equivalent to the following fact: there is convergence of $\mm(t)$ to a unique stationary covariance matrix $\mm_\infty$ for any initial covariance matrix $\mm(0)$. 
\end{lemma}
Note that this lemma already proves thath \ref{as:uniq} implies \ref{crit:spectre}, since the convergence and uniqueness for states implies the convergence and uniqueness for covariance matrices. 

\textit{Proof of the lemma.}
\begin{subproof}
 Write $G=-iT_S-1/2 \Theta^* \Theta$. $G$ is real since $\Theta$ and $T_S$ are in $i \R$ (recall that we are in the field representation).  is Let $\rho_1$ and $\rho_2$ be two initial states and let $\mm_1(t)$ and $\mm_2(t)$ be their covariance matrix after a time $t$. Then $\mm_1(t)-\mm_2(t)$ is the form $iR(t)$ where $R(t)$ is a real antisymmetric matrix and by Formula \ref{eq:mm}, 
 \[
 \frac{d}{dt} R(t)=G R(t)+R(t) G^*.
 \]
Hence, we are reduced to the study of the spectrum of $L: R \mapsto GR+R G^*$. 

Let $\sigma(G)\subset\C$ be the spectrum of $G$. Then the spectrum of $L$ is $\sigma(L)=\sigma(G)+\overline{\sigma(G)}$. 

The key point is the following: all elements of $\sigma(G)$ are of negative real part if and only if Assertion \ref{crit:spectre} is satisfied. 
\begin{subproof}
Indeed, if there exists $u \neq 0\in \C^{2L}$ with $\Theta^* u=0$ and $T_S u=\lambda u$, then $Gu=-i\lambda u$ so $i\lambda \in \sigma(G)$; conversely, if there exists $u\geq 0$ with $Gu=(r+i\lambda) u$ where $r,\lambda \in \R$ and $r>0$, then 
\[
\bra{u} G \ket{u}=(r+i\lambda) \norm{u}^2=i \bra{u}T_S \ket{u}-\frac{1}{2} \norm{\Theta^* u}^2
\]
so by identifing the real and imaginary part, $r\norm{u}^2=-\frac{1}{2} \norm{\Theta^* u}^2\leq 0$, so $r=0$, $\Theta^* u=0$ and $T_Su=i\lambda u$. This contradicts Assertion \ref{crit:spectre}.  
\end{subproof}

Now, let us assume that Assertion \ref{crit:spectre} is not satisfied. Then let $u \neq 0$ be such that $Gu=i\lambda u$. Let $a=\re u$, $b=\im u$, then $Ga =-\lambda b$ and $Gb=\lambda a$. Define $R=\ket{a}\bra{b}-\ket{b}\bra{a}$. Then $R$ is an antisymmetric real matrix and $L(R)=0$. Thus, if we take $\rho_1(0)$ to be a nondegenerate quasi-free state, $\mm_1(0)+i\varepsilon R$ is a covariance matrix for $\varepsilon$ small enough; take $\rho_2(0)$ to be the quasi-free state of covariance matrix $\mm_1(0)+i\varepsilon R$. Then $\mm_1(t)-\mm_2(t)=\varepsilon R$ is nonzero and constant, so $\mm_1(t)$ and $\mm2(t)$ cannot converge towards the same limit.
\end{subproof}

If we were only concerned by states that are even operators, we could stop here. Indeed, if there is convergence and uniqueness of the density matrix, there is convergence and uniqueness for any quasi-free state. But the linear span of quasi-free states is the set of even operators, so with Lemma \ref{lem:critMatrix}, we actually proved that Assertion \ref{crit:spectre} implies the convergence to a unique stationary state for any initial condition that is an even operator. 
\newline

Now, let us deal with initial states with a nonzero odd part. We will make use of the following criteria, from \cite{FagnolaRebolledoAlgebraicCondition}

\begin{prop}\label{prop:FR}
Let $\big(\Lambda^t\big)$ be a quantum dynamical semigroup on a finite-dimensional space $\hh$, with generator $\call(\rho)=-i[K, \rho]+\Phi^*(\rho)$, where $K=.-iH-1/2 \Phi(Id)$ and $\Phi^*(\rho)=\sum_{i=1}^n L_i \rho L_i^*$ for some operators $L_i$.  Assume that there exists a faithfull stationary state $\rho_\infty$. Define the fixed algebra $\ff(\Lambda)$ as the commutator of the set $\set{K, L_i, L_i^*, i=1... n}$ and the decoherence-free algebra $\nn(\Lambda)$ as the commutator of $\set{\left(\delta_H\right)^l(L_i), \left(\delta_H\right)^l(L_i^*), l=0, ..., \infty,\, i=1,...,n}$, where $\delta_H(A)=[H, A]$. Then there is convergence to a stationary state for any initial state if and only if $\nn(\Lambda)=\ff(\Lambda)$. The stationary state is moreover unique if and only if $\nn(\Lambda)=Id\, \C$. 
\end{prop}
Note that there is not uniqueness of the family of operators $(L_i)_{i=1}^n$ such that $\Phi^*(\rho)=\sum_i L_i \rho L_i^*$, but it can be shown that $\nn(\Lambda)$ and $\ff(\Lambda)$ does not depends on the choice of the $L_i$'s.

Let us put aside the requirement that there exists a faithful stationary state and study the algebra $\nn(\Lambda)$.

\begin{lemma}
If Assertion \ref{crit:spectre} is satisfied, then $\nn(\Lambda)=Id\, \C$.
\end{lemma}
\textit{Proof of the lemma.}
\begin{subproof}
 We will write $\Phi^*$ in the form $\sum L_i \rho L_i^*$ with conveniently chosen $L_i$. Write $\cala(\Lambda)=\nn(\Lambda)'$ the unital algebra generated by the $\left(\delta_H\right)^l(L_i), \left(\delta_H\right)^l(L_i^*)$ (which is independent of the choice of the $L_i$). 
 
Assume that we choose the $E_{BS}$ isomorphism. We will first show that $\cala(\Lambda)$ contains $\varphi(u)=\sum_i u_i \gamma_i$ for any $u \in \ran \Theta$. 
\begin{subproof}

 Recall that by Proposition \ref{prop:evolutionRho} we have $ \Phi^*(\rho)=\sum_{1\leq i, j \leq 2L}  \left(\Theta M_B \Theta^*\right)_{i, j} \gamma_j \rho\, \gamma_i$.  Now, since we are in the Majorana basis, $M_B=\frac{1}{2}\left(i R_B-Id\right)$ where $R_B$ is a real antisymmetric matrix with $-Id \leq iR_B \leq Id$. 
 
The matrix $\Theta M_B \Theta^*$ is semi-definite positive; let $u_1, ..., u_r$ be its eigenvectors with positive eigenvalues $\lambda_1, ..., \lambda_r$. Then 
\[
\Phi^*(\rho)=\sum_{i=1}^r \lambda_i\, \varphi(u_i)\, \rho\, \varphi(\overline{u_i})\,.
\]
Thus, we can write $\Phi^*(\rho)=\sum_i L_i \rho L_i^*$ with $L_i=\sqrt{\lambda_i}\varphi(u_i)$. This implies that $\varphi(u_i)$ and $\varphi(\overline{u_i})$ are in $\cala(\Lambda)$. Thus, for any vector $v$ in the image of $\Theta M_B \Theta^*$, the operators $\varphi(v)$ and $\varphi(\overline{v})$ are in $\cala(\Lambda)$. Thus, $\varphi(\Re(v)) \in \cala(\Lambda)$, but since $\Re(\Theta M_B \Theta^*)=\frac{1}{2} \Theta \Theta^*$ this implies that $\varphi(u) \in \cala(\Lambda)$ for any $u \in \ran \Theta$.
\end{subproof}

 Now, for any $u \in \C^2L$, 
 \[
 \delta_{H_S}(\varphi(u))=[H_S, \varphi(u)]=2\,\varphi(T_Su).
 \]
Thus, for any $v \in V(T_S, \Theta)$, the operator $\Phi(v)$ is in $\cala(\Lambda)$. Since Assertion \ref{crit:kalman} is satisfied, it means that all the $\gamma_i$'s are in $\cala(\Lambda)$ and since $\cala(\Lambda)$ is a unital algebra, this implies that $\cala(\gamma)=\bb(\hh)$, hence $\nn(\Lambda)=Id\,\C$.

Now, if we choose the $E_{SB}$ isomorphism instead, by a similar proof we show that $\gamma_i (-1)^{N_S} \in \cala(\gamma)$ for all $i=1, ..., 2L$. But the $\gamma_i (-1)^N$'s also generate $\bb(\hh)$ as an algebra, so $\nn(\Lambda)=Id\,\C$ as well. 
\end{subproof}

This lemma proves that, provided there is a faithful stationary state, Assertion \ref{crit:spectre} implies Assertion \ref{as:uniq}. There is a last technicality to prove it when there is no faithful stationary state. 

 Up to a Bogoliubov transform, we can decompose $\hh$ as $\Gamma(\hh_A^0)\otimes \Gamma(\hh_C^0)$ as in Proposition \ref{prop:support}, so that the quasi-free stationary state $\rho_\infty$ is of support $\set{\ket{\Omega_A}\bra{\Omega_A}}\otimes \Gamma(\hh_C^0)$ and the restricted semigroup $\Lambda_C^t$ is a semigroup on the fermionic space $\Gamma(\hh_C^0)$ of the same form as $\Lambda$ with some matrices $T_S^C$ and $\Theta^C$.

Let us prove that there is convergence and uniqueness for $\Lambda_C$.
\begin{subproof}
  We already proved that Assertion \ref{crit:spectre} for $\Lambda$ implies that for any quasi-free state $\rho$ on $\hh$, $\Lambda^t(\rho)=\Lambda^t(\rho)$ converges to $\rho_\infty$. This implies the same for $\Lambda_C$, since it is only the restriction of $\Lambda$, which in turn implies that Assertion \ref{crit:spectre} is satisfied for $T_S^C$ and $\Theta^C$ by Lemma \ref{lem:critMatrix}. Since $\Lambda_C$ has a stationary state of full support, this implies convergence and uniqueness for any initial state for $\Lambda_C$. 
\end{subproof}
Thus, there is convergence to $\rho_\infty$ for any initial condition with support in $\set{\ket{\Omega_A}\bra{\Omega_A}}\otimes \Gamma(\hh_C^0)$ .
 
  Now, for any initial state $\rho$, with support possibly larger than $\set{\ket{\Omega_A}\bra{\Omega_A}}\otimes \Gamma(\hh_C^0)$, the covariance matrix of $\rho(t)$ converges to $\mm_\infty$, thus for any $\varepsilon >0$, for any $t$ large enough, $\tr{\rho(t) c^A_i \left(c^A_i\right)^*}=\mm_{i, i}(t)$ is at distance at most $\varepsilon/a$ of $1$. This implies that 
\[
\tr{\bra{\Omega_A} \rho(t) \ket{\Omega_A}}\geq 1-\varepsilon
\]
and so, if $\varepsilon <1/2$, 
\[
\norm{\rho(t)- \ket{\Omega_A}\bra{\Omega_A}\otimes\rho_C(t)  }_1 \leq 3\varepsilon
\]
where $\norm{\bullet}_1$ is the trace norm and
\[
\rho_C(t)=\frac{\bra{\Omega_A} \rho(t) \ket{\Omega_A}}{\tr{\bra{\Omega_A} \rho(t) \ket{\Omega_A}}}\,.
\]

Now for any $s$ sufficiently large, $\Lambda_s(\ket{\Omega_A}\bra{\Omega_A}\otimes\rho_B(t) )$ is at distance at most $\varepsilon$ of $\rho_\infty$, and since  $\Lambda_s$ is a contraction, this implies that $\norm{\rho(t+s)-\rho_\infty}\leq 4\varepsilon$. This ends the proof of the theorem. 
\end{proof}

In the case where stationary states are not unique, convergence can still be achieved under some assumptions:

\begin{theo}\label{theo:convergence}
The following propositon are equivalent:
\begin{enumerate}[label=({\arabic*}')]
\item For every initial state $\rho_S(0)$, $\rho_S(t)$ converges towards a stationary state $\Phi(\rho_S(0))$ (not necessarily unique, nor quasi-free).\label{as:conv}
\item $\ker(\Theta^*)$ contains at most one eigenspace of $T_S$. \label{crit:convSpectre}
\item The orthogonal of $V(T_S, \Theta)=Vect\left( \cup_{k=0}^{2L}{Ran }(\,(T_S)^k\, \Theta\,\right)$ is an eigenspace of $T_S$.\label{crit:convKalman}
\end{enumerate}
\end{theo}

This theorem can be proved exactly as Theorem \ref{theo:uniqueness}.

\subsection{The case of systems that are real in the creation/annihilation basis}\label{subsec:realcase}

The criterion can be simplified in the case where $T_S$ and $\Theta$ are real in the creation/annihilation basis. Indeed, these operators are real in the creation/annihilation basis if and only if they are of the following form in the Majorana basis:
\begin{align*}
T_S&=\begin{pmatrix}
0 & iC_T \\
-iC_T^* & 0
\end{pmatrix} &
\Theta&=\begin{pmatrix}
0 & i C_\Theta \\
-iC_\Theta^* & 0
\end{pmatrix}
\end{align*}
where $C_T$ is a real $L\times L$ matrix and $C_\Theta$ is a real $L \times K$ matrix. Thus,
\begin{align*}
T_S^{2k} \Theta&=i\begin{pmatrix}  0& \left(C_TC_T^*\right)^k C_\Theta \\ -\left(C_T^*C_T\right)^k C_\Theta \end{pmatrix} \\ 
T_S^{2k+1}\Theta&=\begin{pmatrix} \left(C_TC_T^*\right)^{k}C_T C_\Theta &0 \\0& \left(C_T^*C_T\right)^{k}C_T^*C_\Theta \end{pmatrix}\,.
\end{align*} 
Hence, Kalman's criterion is satisfied if and only if  
\begin{align*}
Vect\left(\cup_{k=0}^L \text{Ran }(\left(C_TC_T^*\right)^k C_\Theta\right)+Vect\left(\cup_{k=0}^L \left(C_TC_T^*\right)^{k}C_T C_\Theta\right) &=\C^L\\
\intertext{and}
Vect\left(\cup_{k=0}^L \text{Ran }(\left(C_T^*C_T\right)^k C_\Theta\right)+Vect\left(\cup_{k=0}^L \left(C_T^*C_T\right)^{k}C_T^* C_\Theta\right) &=\C^L\,.
\end{align*}

\subsection{The case of gauge-invariant systems}
The results presented above can be simplified in the case of gauge-invariant systems, where every Hamiltonian considered are second quantizations of operators and $\omega$ is itself gauge-invariant. In this section, we use the creation/annihilation basis because it is the most adapted to gauge-invariant systems.

We say that the model is gauge-invariant if $[N_B, \omega]=0$, $[N_S+N_B, V]=0$ and $[N_S, H_S]=0$. It means that we can write $H_S=C_S^* T_{S,0} C_S$, where $T_{S,0}$ is a $L\times L$ hermitian matrix and $V=C_S^* \Theta_0 C_B+C_B^* \Theta_0^* C_S$. Then, with the notation used in the previous sections, we have (in the creation/annihilation basis):
\begin{align*}
T_S&=\frac{1}{2} \begin{pmatrix} T_{S,0} & 0 \\ 0 & -\overline{T_{S,0}} \end{pmatrix} \\
\Theta&=\begin{pmatrix} \Theta_0 & 0 \\ 0 & -\overline{\Theta_0} \end{pmatrix} \\
M_{B}&=\begin{pmatrix} M_{B,0} & 0 \\ 0 & Id-\overline{M_{B,0}} \end{pmatrix}
\end{align*}
Hence, if we write the full covariance matrix of the system as 
\[
\mm(t)=\begin{pmatrix} \mm_0(t) & A_0(t) \\ -\overline{A_0(t)} & Id-\overline{\mm_0(t)} \end{pmatrix}
\]
then Equation \ref{eq:mm} can be decomposed into the following system: 
\begin{align}
\frac{d}{dt} \mm_0(t)&=-i[T_{S,0}, \mm_0(t)]-\frac{1}{2} \set{\Theta_0 \Theta_0^*, \mm_0(t)}+\Theta_0 M_{B, 0} \Theta_0^*\label{eq:mm0} \\
\frac{d}{dt}A_0(t) &=-i\left(T_{S,0} A_0(t)+A_0(t) \overline{T_{S, 0}}\right)+\frac{1}{2}\left(\Theta_0\Theta_0^* A_0(t)-A_0(t) \overline{\Theta_0\Theta_0^*}\right)\,.
\end{align}
If $\rho(0)$ is gauge-invariant, then $\rho(t)$ is gauge-invariant for all $t$. This implies that there exists a quasi-free gauge-invariant stationary state. Moreover, to check the uniqueness and convergence, it is sufficient to study the Kalman condition for $T_{S,0}$ and $\Theta_0$ only: 

\begin{theo}\label{theo:GIuniqueness}
For gauge-invariant models, the following assertion are equivalent: 
\begin{enumerate}[label=(\arabic*'')]
\item There exist a unique stationary state $\rho_\infty$; it is quasi-free, gauge invariant and $\rho(t) \ap \rho_\infty$ for any initial condition $\rho(0)$. \label{as:GIuniq}
\item $\ker \Theta_0^*$ contains no eigenvectors of $T_{S,0}$. \label{crit:GIspectre}
\item The space $V(T_{S,0}, \Theta_0)=\bigoplus \text{Ran }(\,(T_{S,0})^k\, \Theta_0\,)$ is the whole $\C^{L}$.\label{crit:GIkalman}
\end{enumerate}
\end{theo}

Indeed, it is easy to prove that Assertion \ref{crit:GIkalman} is equivalent to Assertion \ref{crit:kalman} of Theorem \ref{theo:uniqueness} in the case of gauge-invariant systems. Note that if these assertions are verivied, then $A_0(t)$ converges to $0$ as $t\ap \infty$ and $\mm_0(t)$ converges to a solution of $-i[T_{S,0}, \mm_0]-\frac{1}{2} \set{\Theta_0 \Theta_0^*, \mm_0}+\Theta_0 M_{B, 0} \Theta_0^*=0$.

Theorem \ref{theo:convergence} also has an obvious simplification in the gauge-invariant case. 
\newline

As a conclusion, the study of gauge-invariant models is simplified: it is sufficient to study Equation \ref{eq:mm0} to study the convergence and to find the covariance matrix of the stationary state. 

Gauge-invariant models provides many interesting examples, as shown in the following section.

\section{Applications and examples in the gauge-invariant case}

In this section, we study a few examples of gauge-invariant systems that are exactly solvable. The first example demonstrate thermalisation (i.e. convergence to a Gibbs state for $H_S$). The second examples shows that thermalisation does not occurs in general, even when the bath is simple. In the third example, we study a model of fermionic chain, which can be seen as a spin chain as shown in Section \ref{sec:XYchain}.

\subsection{Thermalisation}
This example is similar to a model studied by Attal and Joye in \cite{AttalJoye}. It shows how a system may be forced to thermal equilibrium, but the setup is somewhat artificial:

We take the system $\hh_S$ to be a general fermionic space (of finite dimension) and $\hh_B$ a copy of $\hh_S$, with $T_{B,0}=T_{S,0}$. In each interactions, he systems interacts site by site, so $\Theta_0=I$. We suppose that the bath's subsistems are at thermal equilibrium, i.e.
\[
\omega=\rho_\beta=\frac{1}{Z_\beta} e^{-B^* T_{B,0} B}\,.
\]

Clearly, $\ker \Theta_0^*=0$ so the criteria of convergence and uniqueness of Theorem \ref{theo:uniqueness} is satisfied. Moreover, $\mm_\infty^0=M_B$ is a solution of the equation for a stationary state
\[
i[T_{S,0}, \mm^0]+\frac{1}{2}\set{I, \mm^0}= M_B^0\,.
\]
Hence $\rho_S=\rho_\beta$ is the unique stationary state: there is convergence to the thermal equilibrium at temperature $\beta$ for any inital condition $\rho(0)$.

\subsection{A simple bath}\label{subsec:onesite}
Let us assume that the bath is a Gibbs state at temperature $\beta$ with the most simple possible Hamiltonian: $H_B=N_B=\sum_{i=1}^K \left(c_{i}^B\right)^* c_i^B$. Then 
\[
M_{B,0}=\left(1-e^{-\beta}\right)Id_K 
\]
Then, writing Equation \ref{eq:mm0}, we see that the gauge-invariant quasi-free state with small covariance matrix 
\[
\mm_0=\left(1-e^{-\beta}\right)Id_L 
\]
is a stationary state. Note that it does not depends on $T_S$ or $\Theta$, though it commutes with $H_S$. This is somewhat counter-intuitive, since we would expect simple systems to thermalise for $T_S$. This is an effect of the high degeneracy of $H_S$ (it commutes with many non-trivial operators, including $N_S$). 
\newline
Here are two special cases:

\subsubsection{A fermionic chain interacting at one end}
Here is an example where the convergence and uniqueness criteria is satisfied: we take $\dim \hh_{B, 0}=1$ and the system is a chain of fermions interacting with the bath in one end of the chain. 
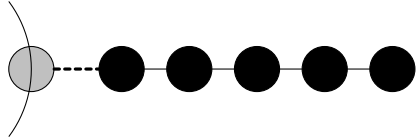
\begin{figure}[h]
\begin{center}
\begin{tikzpicture}[line cap=round,line join=round,x=0.3cm,y=0.3cm]
\draw [line width=0.4pt,fill=black,fill opacity=1.0] (-4,0) circle (0.3cm);
\draw [fill=black,fill opacity=1.0] (-7,0) circle (0.3cm);
\draw [fill=black,fill opacity=1.0] (-10,0) circle (0.3cm);
\draw [fill=black,fill opacity=1.0] (-13,0) circle (0.3cm);
\draw [fill=black,fill opacity=1.0] (-16,0) circle (0.3cm);
\draw [line width=0.4pt,fill=black,fill opacity=0.25] (-20,0) circle (0.3cm);
\draw [shift={(-25,0)}] plot[domain=-0.64:0.64,variable=\t]({1*5*cos(\t r)+0*5*sin(\t r)},{0*5*cos(\t r)+1*5*sin(\t r)});
\draw [line width=1.2pt,dash pattern=on 2pt off 2pt] (-18.99,0)-- (-17,0);
\draw (-15,0)-- (-5,0);
\end{tikzpicture}
\end{center}
\caption{A fermionic chain with one bath}
\end{figure}

Let us define 
\begin{align}
D=\begin{pmatrix}
0 & 1 &  & \\
0 & 0 & 1 & & 0\\
& 0 & 0 & 1 & \\
 & &\ddots & \ddots & \ddots 
\end{pmatrix} \label{eq:D}
\end{align}
and chose
\begin{align*}
T_{S,0}&=D+D^T & 
\Theta_0^*&=\begin{pmatrix} 1 & 0 & ... & 0 \end{pmatrix}\,.
\end{align*}

 Then $\left(T_S\right)^k \Theta$ is the form 
\[
\begin{pmatrix}
*\\
\vdots\\
1\\
0\\
\vdots
\end{pmatrix}
\]
where $*$ represents the $k$ first entries. Thus, Kalman's criterion is satisfied. 

\subsubsection{The star}

Now, here is an example where the criteria of convergence is satisfied but not the criteria of uniqueness: whe take $\hh_S$ to be a "star", with every site connected to a central site which interacts with the bath:
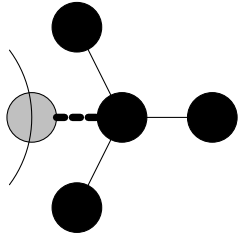
\begin{figure}[!h]
\begin{center}
\begin{tikzpicture}[line cap=round,line join=round,x=0.3cm,y=0.3cm]
\draw [fill=black,fill opacity=1.0] (-16,0) circle (0.33cm);
\draw [fill=black,fill opacity=1.0] (-18,4) circle (0.33cm);
\draw [fill=black,fill opacity=1.0] (-12,0) circle (0.33cm);
\draw [fill=black,fill opacity=1.0] (-18,-4) circle (0.33cm);
\draw [fill=black,fill opacity=0.25] (-20,0) circle (0.33cm);
\draw (-18,4)-- (-16,0);
\draw (-16,0)-- (-18,-4);
\draw (-16,0)-- (-12,0);
\draw [line width=2.8pt,dash pattern=on 3pt off 3pt] (-18.9,0)-- (-16,0);
\draw [shift={(-25,0)}] plot[domain=-0.64:0.64,variable=\t]({1*5*cos(\t r)+0*5*sin(\t r)},{0*5*cos(\t r)+1*5*sin(\t r)});
\end{tikzpicture}
\end{center}
\caption{A star-shaped system}
\end{figure}

We can for example define $T_{S,0}$ and $\Theta_0$ the following way: 
\begin{align*}
T_{S,0}&=\begin{pmatrix}
0 & 1 & \cdots &1 \\
1 & \\
\vdots &  & \scalebox{2}{0} \\
1
\end{pmatrix}
&
\Theta_0^*&= \begin{pmatrix} 1 & 0 & ... & 0 \end{pmatrix}.
\end{align*}

Then $\ker T_{S,0}$ is of dimension $L-1$ so $\Theta_0$ cannot be cyclic. However, its non zero eigenvalue is $\sqrt{n}$ with eigenvector
\[
\begin{pmatrix}		
\sqrt{n}\\		
1\\
\vdots\\
1
\end{pmatrix}
\] 
which is not in $\ker \Theta_0^*=\set{0}\times \C^{L-1}$, so the criteria of convergence from theorem \ref{theo:convergence} is satisfied: there is convergence to a (non-unique) stationary state for every initial state.

\subsection{The fermionic chain with interaction at both ends} \label{fermionic_chain}

We consider two bath interacting at each ends of a chain of sites with nearest neighbors interactions: 
\begin{figure}[h]
\begin{center}
\begin{tikzpicture}[line cap=round,line join=round,x=0.3cm,y=0.3cm]
\draw [line width=0.4pt,fill=black,fill opacity=0.25] (0,0) circle (0.3cm);
\draw [line width=0.4pt,fill=black,fill opacity=1.0] (-4,0) circle (0.3cm);
\draw [fill=black,fill opacity=1.0] (-7,0) circle (0.3cm);
\draw [fill=black,fill opacity=1.0] (-10,0) circle (0.3cm);
\draw [fill=black,fill opacity=1.0] (-13,0) circle (0.3cm);
\draw [fill=black,fill opacity=1.0] (-16,0) circle (0.3cm);
\draw [line width=0.4pt,fill=black,fill opacity=0.25] (-20,0) circle (0.3cm);
\draw [shift={(5,0)}] plot[domain=2.5:3.79,variable=\t]({1*5*cos(\t r)+0*5*sin(\t r)},{0*5*cos(\t r)+1*5*sin(\t r)});
\draw [shift={(-25,0)}] plot[domain=-0.64:0.64,variable=\t]({1*5*cos(\t r)+0*5*sin(\t r)},{0*5*cos(\t r)+1*5*sin(\t r)});
\draw [line width=1.2pt,dash pattern=on 2pt off 2pt] (-18.99,0)-- (-17,0);
\draw [line width=1.2pt,dash pattern=on 2pt off 2pt] (-1,0)-- (-3,0);
\draw (-15,0)-- (-5,0);
\end{tikzpicture}
\end{center}
\caption{A fermionic chain with two baths}
\end{figure}
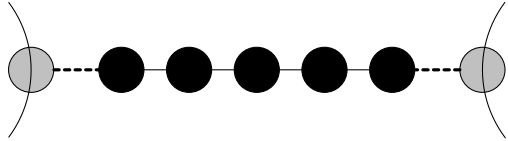

Thus, we take $K=2$ and
\begin{align*}
M_{B,0}&=\begin{pmatrix} n_1 & 0 \\ 0 & n_L \end{pmatrix} \\
\Theta_0&=
\begin{pmatrix}
\theta_1 & 0 \\
0 & 0 \\
\vdots & \vdots \\
0 & \theta_L
\end{pmatrix}\\
T_{S,0}&= D+D^T
\end{align*}
where $D$ is the upper-diagonal matrix defined in \ref{eq:D} and $0 \leq n_1, n_L \leq 1$, with $\theta_1, \theta_L \in (0, \infty]$. 

 This example has been solved by Karevski and Platini in \cite{KarevskiPlatini2009} in the case $\theta_1=\theta_L=1$ and what follows is just a generalisation of their computation in this slightly more general context. We are able to fully describe the stationary state. 

\begin{prop}
There is convergence to a unique stationary state, whose small covariance matrix is the form
\[
\mm_0=\begin{pmatrix}
p_1 & 0 & 0 & ... \\
0 & p_m &0 ... \\
0& 0 & p_m & ... \\
&&& \ddots& \\
0 & ...& &&p_m \\
0 & ... &&&& p_L
\end{pmatrix}+i c (D-D^T)
\]
where $p_1, p_m, p_L $ and $c$ are real numbers that are independent of $L$. They are defined as follows.
Let $s=4(\theta_1^2+\theta_L^2)+\theta_1^2\theta_L^2(\theta_1^2+\theta_L^2)$. Then 
\begin{align*}
p_1&=\frac{1}{s}\Big(\left(\theta_1^2\theta_L^4+\theta_1^4\theta_L^2+4\theta_1^2\right)n_1 +4\theta_L^2n_L\Big) \\
p_m&=\frac{1}{s}\Big(\theta_1^2\left( \theta_L^4+4\right)n_1+\theta_L^2\left(\theta_1^4 +4\right)n_L\Big) \\
p_L&= \frac{1}{s}\Big( 4\theta_1^2n_1+\left(\theta_L^2\theta_1^4+\theta_L^4\theta_1^2+4\Theta_L^4\right)n_L\Big)\\
c&= \frac{2}{s} \theta_1^2\theta_L^2 (n_1-n_L) \,.
\end{align*}
\end{prop}

{\bf Remarks:} 
\begin{enumerate}
\item The numbers $p_1, p_m, p_L$ are the average number of particles in their respective sites. As expected, they are barycenters of $n_1$ and $n_2$, the average number of particles in each bath. If $\theta_1=\theta_2$, the average number of particles in any site exept the first and the last of the chain is $p_m=(n_1+n_2)/2$. But if we take different values for $\theta_1$ and $\theta_2$, we may observe a strange behavior: if $-\sqrt{2}< \theta_1 <\sqrt{2}$ then increasing slightly $\theta_1$ makes $p_m$ closer to $n_L$. Indeed, $p_m=a(\theta_1, \theta_2)n_1+(1-a(\theta_1, \theta_2))n_L)$ where $a(\theta_1, \theta_2)=\theta_1^2\left( \theta_L^4+4\right)/s$ and if $\theta_1^2<0$ then 
\[
\frac{\partial}{\partial \theta_1} \left(\frac{a(\theta_1, \theta_2)}{1-a(\theta_1, \theta_2)}\right) <0\,.
\]

\item The number $c$ is the current of particles. It is independent of $L$, which shows that the system does not verify any kind of Fourier's law. This is not a surprise, because harmonics systems with invariance by translation as our are known to contradicts Fourier's law.
\end{enumerate}

\begin{proof}
Similarly to the fermionic chain interacting at one end, it is easy to show that Kalman's criterion is satisfied, so there is a unique stationary state. Hence it is sufficient to show that the small covariance matrix described in the proposition is a solution of 
\[
-i[T_{S,0}, \mm_0]-\frac{1}{2} \set{\Theta_0 \Theta_0^*, \mm_0}+\Theta_0 M_{B, 0} \Theta_0^*=0\,.
\]
Let $\mm_0=P+ic(D-D^T)$ be the covariance matrix described in the proposition. Write 
\begin{align*}
a_1&=-2c+\theta_1^2(n_1-p_1) & a_3&= p_1-p_m-c\frac{\theta_1^2}{2}\\
a_2&=2c+\theta_L^2(n_L-p_L) & a_4&= p_L-p_m+c\frac{\theta_L^2}{2}\,.
\end{align*}

Then a simple computation shows that the real part of the equation for $\mm_0$ is
\begin{align*}
-i[T_{S0}, ic(D-D^T)]-\frac{1}{2}\set{\Theta_0\Theta_0^*, P}+\Theta_0M_{B,0} \Theta_0^*&= 
\begin{pmatrix} a_1 & 0 & ... & 0\\
0 & 0 & ... \\
\vdots&& \ddots & 0 \\
0&...&0& a_2\\
\end{pmatrix}
\end{align*}
and the imaginary part is 
\begin{align*}
-i[T_{S0}\, , \, P]-\frac{1}{2} \set{\Theta_0\Theta_0^* \, , \, ic(D-D^T)}&=
 i \begin{pmatrix} 
 0 & -a_3 & 0 & ...&0 \\
 a_3 &0 \\
 0&  &\ddots & & \\
\vdots &  & &0 &a_4 \\
0 &      & &-a_4 & 0
\end{pmatrix}\,.
\end{align*}
Thus, $\mm_0$ is solution if and only if $a_1=a_2=a_3=a_4=0$. But $p_1, p_m, p_L$ and $c$ are precisely the solution of these equations, so $\mm_0$ is the covariance matrix of the stationary state. 
\end{proof}

{\bf Remark:} Note that the simple form of the solution is due to the invariance by translation of $T_S$. If we disturb $H$ a little (by example, we add $\varepsilon \ket{i}\bra{i}$ with $\varepsilon >0$ for some fixed $i \in {1,..., L}$), the solution is no more tri-diagonal.

\section{The XY-model} \label{sec:XYchain}

The one-dimensional XY model give rise to an important example of non Gauge-invariant fermionic system. The closed system version of this model was solved in \cite{Lieb1961}, \cite{Barouch1970a1971}. The continuous-time limit repeated interaction process applied on the XY-model is studied in Dharhi \cite{Dharhi2008}, where the convergence to a unique stationary state was proved in the case where the baths where in a non-degenerate state. The explicit form of the stationary state n the anisotropic case was derived in \cite{Platini2008}. It is also studied in \cite{Prosen2010}, in the context of quantum semigroups, without underlying repeated interaction process. 

\subsection{Description of the XY-model}\label{subsec:JordanWigner}

The XY-model is a chain of $L$ spin interacting with their nearest neighbors; it is represented by $\hh=\left(\C^2\right)^{\otimes L}$ with Pauli matrices $\sigma_n^x, \sigma_n^y, \sigma_n^z$ (with $\sigma_n^\gamma$ commuting with $\sigma_m^\beta$ for $n\neq m$.) The Hamiltonian is defined as
\begin{align}\label{xy}
H= -\frac{1}{2}\sum_{n=1}^{L-1}\left( \frac{1+\kappa}{2} \sigma_n^x \sigma_{n+1}^x + \frac{1-\kappa}{2} \sigma_n^y \sigma_{n+1}^y \right) -\frac{h}{2} \sum_{n=1}^L \sigma_z\,,
\end{align}
where $\kappa \in [0,1]$ (the choice of $(1+\kappa)/2$ and $(1-\kappa)/2$ as the coupling constant for $\sigma^x$ and $\sigma^y$ is general and will simplify the notations).

We will study a version of this model where it in interaction at both ends with two spin system: we take the continuous-time limit of repeated interactions, with $H_S$ taken as above and 
\[
V=\frac{1}{2}\theta_1\left( \frac{1+\kappa}{2} \sigma_{B1}^x \sigma_{1}^x + \frac{1-\kappa}{2} \sigma_{B1}^y \sigma_{1}^y \right)  +
\frac{1}{2}\theta_1\left( \frac{1+\kappa}{2} \sigma_{L}^x \sigma_{BL}^x + \frac{1-\kappa}{2} \sigma_{L}^y \sigma_{BL}^y \right)\,.
\]
where $\theta_1, \theta_2 \in (0, \infty)$. This model can be mapped to a fermionic model with the help of the Jordan-Wigner transform. 

\subsection{The Jordan-Wigner transform}

The setup of the Jordan-Wigner transform is the following: let $\ket{u}, u \in \{-1, 1\}^L$ be the "up z, down z" basis of $\hh$, characterized by $\sigma_n^z\ket{u}=u_n \ket{u}$. We define fermionic operators on $\hh$ the most natural possible way:
\[
c_i \ket{u}=\left \{
\begin{array}{l}
0 \text{ if } u_i=0 \\
\prod_{k=1}^{i-1} (-1)^{u_i} \ket{u_1, \cdots, u_i-1, \cdots u_L} \text{ if } u_i=1
\end{array} \right.\,.
\]
The operators $c_i, c_i^*$ make $\hh$ a fermionic space and they are related to the $\sigma_n^\gamma$ by many relations: define

\begin{align}
\sigma_n^{\pm}&=\frac{1}{2}(\sigma_n^y\pm i \sigma_n^y) \\
A_n&=\Pi_{j=1}^{n-1} (-\sigma_j^z)
\end{align}
then
\begin{align}
\sigma_n^z&=c_n^* c_n- c_n c_n^*=\frac{i}{2}(\gamma_n\gamma_{n+L}-\gamma_{n+L}\gamma_{n}) \\
\sigma_n^x&=- A_n (c_n^*+ c_n)=A_n\gamma_n\\
\sigma_n^y&= i A_n(c_n^*- c_n)=A_n\gamma_{n+L}\,.
\end{align}
and thus
\begin{align}
\sigma_n^x \sigma_{n+1}^x&=(c_n^*-c_n)(c_{n+1}^*+c_{n+1})=\frac{i}{2}(\gamma_{n+L}\gamma_{n+1}-\gamma_{n+1}\gamma_{n+L}) \\
\sigma_n^y \sigma_{n+1}^y&=- (c_n^*+c_n)(c_{n+1}^*-c_{n+1})=\frac{i}{2}(\gamma_n\gamma_{n+1+L}-\gamma_{n+1+L}\gamma_n)\,.
\end{align}

Note that this transformation is dependent in the way we ordered the subsystems of $\hh=(\C)^{\otimes L}$.

We see that the Hamiltonian of the XY-model writes
\[
H=\frac{1}{2}F^* T F+ c\,Id
\]
where $T$ is a $2L\times 2L$-matrix; in the creation-annihilation, it is the form 
\[
T=\begin{pmatrix}
A & B \\
-B & -A\\
\end{pmatrix}\,.
\]
To describe $A$ and $B$, let $D$ be the matrix with 1 on the upper diagonal and 0 elsewhere; then
\begin{align}
A&=- h\,Id -\frac{1}{2} (D+D^T)\\
B&=-\kappa/2 (D-D^T)\,.
\end{align}
In the Majorana basis, it is the form
\[
T=\frac{1}{2}\begin{pmatrix}
0 & i C \\
-iC & 0
\end{pmatrix}
\]
where 
\[
C=h\,Id+\frac{1-\kappa}{2} D+\frac{1+\kappa}{2} D^T.
\]

This Hamiltonian can then by diagonalized explicitely with the method of Subsection \ref{subsec:quadratic_hamiltonian}, but it is not needed here. 

Note that if we include longer-range interactions in $H_S$, like $\sigma_n^x \sigma_{n+2}^y$, the Hamiltonian obtained under the Jordan-Wigner transform is no longer quadratic and the following study breaks down.

\subsection{Repeated interaction on the XY-model}

Let us consider a spin chain in repeated interaction with two spin chains by its ends; since it is only in contact with the ends of these chains, we can assume they have only one spin each, thus the space of the bath will be $\hh_B=\hh_{B1}\otimes \hh_{B1}$ where $\hh_{B1}=\C^2$ represents the spin on the left and $\hh_{B2}=\C^2$ represents the state on the right.  We will show that there is convergence and uniqueness of the stationary state in the case $\kappa \neq 1$ or $h\neq 0$. Let $\sigma^\alpha_n$ be the $\alpha$-Pauli matrix for the $n$-th spin as in Subsection \ref{subsec:JordanWigner} and let $\sigma^\alpha_{k,B}$ be the $\alpha$-Pauli matrix for the bath (with $k=1,2$). The Hamiltonian of the system is $H_S=H$ as defined in Subsection \ref{subsec:JordanWigner} and we choose the interaction Hamiltonian to be
\[
V=-\frac{\theta_1}{2}\left(\frac{1+\kappa}{2} \sigma_{1,B}^x \sigma_1^x+\frac{1-\kappa}{2} \sigma_{B1}^y \sigma_1^x \right) -\frac{\theta_2}{2} \left(\frac{1+\kappa}{2} \sigma_{L}^x \sigma_{B2}^x+\frac{1-\kappa}{2} \sigma_{L}^y \sigma_{2,B}^y \right)\,.
\]
Let us consider the Jordan-Wigner transform on $\hh_{B1}\otimes \hh_S\otimes \hh_{B2}$ where the subsystems are in the same order as the chain. Then we can write the interaction in terms of the creation and annihilations operators on $B$ and $S$ :
\[
V=\frac{1}{2}\left(F_S^*  \Theta F_B+F_B^*\Theta^*F_S\right)
\]
where in the creation/annihilation basis, 
\[
\Theta=\begin{pmatrix}
A_\Theta & B_\Theta  \\
-B_\Theta & -A_\Theta
\end{pmatrix}
\]
with
\begin{align*}
A_\Theta&=-\frac{1}{2}\begin{pmatrix}
\theta_1 & 0\\
\vdots & \vdots\\
0 & \theta_L
\end{pmatrix}&
B_\Theta&=-\frac{\kappa}{2} \begin{pmatrix}
-\theta_1 & 0 \\
\vdots & \vdots \\
0 & \theta_L
\end{pmatrix}\,.
\end{align*}
In the Majorana basis, we have
\[
\Theta=\begin{pmatrix}
0 & iC_\Theta \\
-iC_\Theta & 0
\end{pmatrix}
\]
where 
\[
C_\Theta=-\frac{1}{2}\begin{pmatrix}
(1+\kappa)\theta_1 & 0\\
\vdots & \vdots\\
0 & (1-\kappa)\theta_L
\end{pmatrix}
\]
Note that since we are forced to put the tensor products in the order
 $\hh_{B1}\otimes \hh_{S}\otimes \hh_{B2}$,
  we are not using one of the standard isomorphisms $E_{SB}$ and $E_{BS}$ between $\hh_S\otimes \hh_B=\Gamma(\hh^0_S)\otimes \Gamma(\hh_S^0)$ and $\Gamma(\hh^0_B\oplus \hh_S^0)$
   but the isomorphism $E_{B1SB2}=E_{B1(SB2)}\circ E{SB2}$ (see subsection \ref{subsec:tensorProduct}).
    Thus, we have to assume that the state of the bath is of the form $\rho_B=\rho_{B1}\otimes \rho_{B2}$ with $\rho_{B1}$ and $\rho_{B2}$ quasi-free,
     otherwise the state $\rho\otimes \rho_B$ may not be quasi-free even when $\rho$ and $\rho_B$ are quasi-free. The CP map $\Phi^*$ in the quantum semigroup is then the form
     \[
     \Phi^*(\rho)=\sum_{i, j=1}^L [ \Theta_1 M_{B1} \Theta_1^*  ]_{i, j}\, \gamma_j  \rho \gamma_i+  [ \Theta_2 M_{B2} \Theta_2^*  ]_{i, j}\,\gamma_j (-1)^{N_S} \rho(-1)^{N_*S} \gamma_i.
     \]
The theorems of Section \ref{sec:repeated_interaction} are still true in this case. We can then study the convergence and uniqueness properties for this quantum semigroup. 

The following proposition have already been proved in slighltly less general forms in \cite{Dharhi2008} and in \cite{Prosen2010}.

\begin{prop}
In the interacting XY chain as defined above, when $\theta_1$ and $\theta_2$ are nonzero, there is convergence and uniqueness of the stationary state if $\kappa^2 \neq 1$ or $h \neq 0$.
\end{prop} 	
In the isotropic case $\kappa=0$, an explicit formula for the stationary state can be obtained. 

\begin{proof}
We are in the real case (see Subsection \ref{subsec:realcase}). Hence it is sufficient to study the range of $\left(C_TC_T^*\right)^k C_\Theta$ and $\left(C_TC_T^*\right)^kC_T C_\Theta$ on the one hand and of $\left(C_T^*C_T\right)^k C_\Theta$ and $\left(C_T^*C_T\right)^kC_T^* C_\Theta$ on the other hand. 

Let us first study the case where $\kappa^2 \neq 1$. Then $C_T$ is tri-diagonal, so $C_TC_T^*$ has coefficients of indices $(i, i+k)$ equal to zero when $k\geq 2$ and its coefficient of indices $(i, i+2$ are all equals to $(1-\kappa^2)/4$.  Thus, the first column of  $\left(C_TC_T^*\right)^k C_\Theta $ is the form
\[
\begin{pmatrix}
* \\
\vdots  \\
\left(\frac{1-\kappa^2}{4}\right)^k \theta_1  \\
0  \\
\vdots 
\end{pmatrix}
\]
with the last nonzero term in position $2n$. Now, the first column of $\left(C_TC_T^*\right)^kC_T C_\Theta$ is of the same form, but with last nonzero coefficient in position $2n+1$. Hence these vectors generate $\C^L$. Now $C_T^* C_T$ is of the same form  as $C_T C_T^*$, so a similar proof shows that the first columns of $\left(C_T^*C_T\right)^k C_\Theta$ and $\left(C_T^*C_T\right)^kC_T^* C_\Theta$ generate $\C^L$ as well. This proves Kalman's criterion. 
\newline

Now let us assume that $\kappa^2=1$. Up to exchanging $x$ and $y$, we can assume that $\kappa=1$. Then if $h\neq 0$, $C_TC_T^*$ is tri-diagonal with coefficient of indices $(i, i+1)$ equal to $h/4$, so it is easy to prove Kalman's criterion as above.   
\end{proof}

In the case $\kappa=0$ (the XX chain) the matrix $B_T$ and $B_\Theta$ are zero and the system is actually the Gauge-Invariant fermionic chain studied in Subsection \ref{fermionic_chain}. This has been remarked by Karevski and Platini in \cite{KarevskiPlatini2009}. The consequence is that, as for the fermionic chain, the magnetization currents in the XX chain does not depend on the length of the chain.

\nocite{*}
\bibliography{fermions}

\begin{thebibliography}{10}

\bibitem{AlickiLendi}
Robert Alicky and Karl Lendi.
\newblock {\em Quantum dynamical semigroups and applications}.
\newblock Springer Verlag, 1987.

\bibitem{Araki1968}
Huzihiro Araki.
\newblock On the diagonalisation of a bilinear hamiltonian by a bogoliubov
  transformation.
\newblock {\em Publ. RIMS, Kyoto Univ. Ser. A}, 4:387--412, 1968.

\bibitem{Araki1984}
Huzihiro Araki.
\newblock On the xy-model on two-sided infinite chain.
\newblock {\em Publ. RIMS, Kyoto Univ.}, 20:277--296, 1984.

\bibitem{Araki1983}
Huzihiro Araki and Eytan Barouch.
\newblock On the dynamics and ergodic properties of the xy-model.
\newblock {\em J. Stat. Phys.}, 31:407--345, 1983.

\bibitem{Aschbacher2003}
Walter~H. Aschbacher and Claude-Alain Pillet.
\newblock Non-equilibrium steady states of the xy chain.
\newblock {\em J. Stat. Phys.}, 112:1153--1175, 2003.

\bibitem{AttalLecture}
St\'ephane Attal.
\newblock {\em Lectures in quantum noise theory}.

\bibitem{AttalToy}
St\'ephane Attal.
\newblock Approximation of the fock space with the toy fock space.
\newblock {\em Lecture Notes in Math.}, 2003.

\bibitem{AttalOQS}
St\'ephane Attal.
\newblock {\em Open Quantum Systems II}, chapter Quantum Noises.
\newblock Springer, 2005.

\bibitem{AttalJoye}
St\'ephane Attal and Alain Joye.
\newblock The langevin equation for a quantum heat bath.
\newblock {\em Journal of Functional Analysis}, 247(2):253--288, 2007.

\bibitem{AttalLindsay}
St\'ephane Attal and J.~Martin Lindsay.
\newblock Quantum stochastic calculus with maximal operator domains.
\newblock {\em Ann. Probab.}, 32(1A):488--529, 01 2004.

\bibitem{AttalMeyer}
St\'ephane Attal and Paul~Andr\'e Meyer.
\newblock {\em Interpr\'etation probabiliste et extension des int\'egrales
  stochastiques non commutatives}.
\newblock Springer, 1993.

\bibitem{AttalPautrat}
St\'ephane Attal and Yann Pautrat.
\newblock From repeated to continuous quantum interactions.
\newblock {\em Annales Henri Poincar\'e (Physique th\'eorique)}, 7:59--104,
  2006.

\bibitem{OQWbirth}
St\'ephane Attal, Christophe Sabot, F.~Petruccionne, and I.~Sinayskiy.
\newblock Open quantum random walks.
\newblock {\em Phys. Lett. A.}, 376(18):1545, 2012.

\bibitem{B.Dierckx2008}
M.~Fannes B.~Dierckx and M.~Pogorzelska.
\newblock Fermionic quasi-free maps and information theory.
\newblock {\em J. Math. Phys.}, 49:032109, 2008.

\bibitem{BarchielliBelavkin91}
A.~Barchielli and V.P. Belavkin.
\newblock Measurement continuous in time and a posteriori states in quantum
  mechanics.
\newblock {\em J. Phys. A : Math. Gen.}, 1991.

\bibitem{Barouch1970a1971}
Eytan Barouch, Barry~M. McCoy, and \textit{al.}
\newblock Statistical mechanics of the xy model, i, ii, iii, iv.
\newblock {\em Phys. Rev. A}, 1970 to 1971.

\bibitem{spikes3}
Michel Bauer and Denis Bernard.
\newblock Stochastic spikes and strong noise limits of stochastic differential
  equations.
\newblock 05 2017.

\bibitem{bauerBernardTilloy_ballistic}
Michel Bauer, Denis Bernard, and Antoine Tilloy.
\newblock Open quantum random walks: Bistability on pure states and
  ballistically induced diffusion.
\newblock 88, 03 2013.

\bibitem{spikes2}
Michel Bauer, Denis Bernard, and Antoine Tilloy.
\newblock Zooming in on quantum trajectories.
\newblock 49, 12 2015.

\bibitem{Belavkin92}
V.~P. Belavkin.
\newblock Quantum continual measurements and a posteriori collapse on ccr.
\newblock {\em Communications in Mathematical Physics}, 146(3):611--635, Jun
  1992.

\bibitem{Belavkin94}
V.~P. Belavkin.
\newblock Nondemolition principle of quantum measurement theory.
\newblock {\em Foundations of Physics}, 1994.

\bibitem{BJPP}
T.~Benoist, V.~Jak{\v{s}}i{\'{c}}, Y.~Pautrat, and C.-A. Pillet.
\newblock On entropy production of repeated quantum measurements i. general
  theory.
\newblock {\em Communications in Mathematical Physics}, 357(1):77--123, Jan
  2018.

\bibitem{Berezin1966}
F.A. Berezin.
\newblock {\em The Method of second quantization}, volume~24.
\newblock Academic press, 1966.

\bibitem{BhatSinha92}
Rajarama Bhat and Sinha.
\newblock A stochastic differential equations with time-dependent and unbounded
  operator coefficients.
\newblock {\em Journal of Functional Analysis}, 1993.

\bibitem{Bhati}
Rajendra Bhatia.
\newblock {\em Matrix Analysis}, volume 169.
\newblock Springer.

\bibitem{Bhatia1997}
Rajendra Bhatia and Peter Rosenthal.
\newblock How and why to solve the operator equation ax-xb=y.
\newblock {\em Bull. London Math. Soc}, 29:1--21, 1997.

\bibitem{Bogachev}
Vladimir~I. Bogachev.
\newblock {\em Measure Theory, I, II, III}.
\newblock Springer, 2006.

\bibitem{Breuer2002}
H.-P Breuer and F.~Petruccionne.
\newblock {\em The theory of open quantum systems}.
\newblock Oxford University Press, 2002.

\bibitem{brislawn91}
Chris Brislawn.
\newblock Traceable integral kernels on countably generated measure spaces.
\newblock {\em Pacific Journal of Mathematics}, 150(2):229--240, 1991.

\bibitem{cooling_entanglement_bhlp}
Nicolas Brunner, Marcus Huber, Noah Linden, Sandu Popescu, Ralph Silva, and
  Paul Skrzypczyk.
\newblock Entanglement enhances cooling in microscopic quantum refrigerators.
\newblock {\em Physical Review E: Statistical, Nonlinear, and Soft Matter
  Physics}, 89(3), 3 2014.

\bibitem{Burgarth13}
D~Burgarth, G~Chiribella, V~Giovannetti, P~Perinotti, and K~Yuasa.
\newblock Ergodic and mixing quantum channels in finite dimensions.
\newblock {\em New Journal of Physics}, 15:073045, 2013.

\bibitem{Carmichael98}
Howard Carmichael.
\newblock {\em Statistical Methods in Quantum Optics I : Master Equations and
  Fokker-Plank Equations}.
\newblock Springer, 1998.

\bibitem{Cheong2004}
Siew-Ann Cheong and Christopher~L. Henley.
\newblock Many body density matrix for free fermions.
\newblock {\em Physical Review B}, 69, 2004.

\bibitem{briefHistory}
Dariusz Chruściński1 and Saverio Pascazio.
\newblock A brief history of the gksl equation.
\newblock {\em Open systems and Information dynamics}, 24(03), September 2017.

\bibitem{Chung2001}
Ming-Chiang Chung and Ingo Peschel.
\newblock Density-matrix spectra of solvable fermionic systems.
\newblock {\em Physical Review B}, 64:064412, 2001.

\bibitem{Davies74}
Edward~Brian Davies.
\newblock Markovian master equations.
\newblock {\em Commun. Math. Phys}, 39, 1974.

\bibitem{DerezinskiGerard}
Jan Derezinski and Cristian Gerard.
\newblock {\em Mathematics of Quantization and Quantum Fields}.
\newblock Cambridge Monographs on Mathematical Physics, 2013.

\bibitem{Derezinski2007}
Jan Derezinski and Wojciech~De Roeck.
\newblock Extended weak coupling limit for friedrichs hamiltonians.
\newblock {\em Journ. Math. Physics}, 2007.

\bibitem{Dharhi2008}
Ameur Dharhi.
\newblock A lindblad model for a spin chain coupled to heat baths.
\newblock {\em J. Phys. A: Math. Theor.}, 41:275305, 2008.

\bibitem{FagnolaRebolledo2002}
F.~Fagnola and R.Rebolledo.
\newblock Lectures on the qualitative analysis of quantum markov semigroups.
\newblock {\em Quantum probability and white noise analysis}, XV:197--240.,
  2002.

\bibitem{FagnolaUnbounded90}
Franco Fagnola.
\newblock On quantum stochastic differential equations with unbounded
  coefficients.
\newblock {\em Probability Theory and Related Fields}, 86:501--516, 1990.

\bibitem{SubharmonicProjections}
Franco Fagnola and Rolando Rebolledo.
\newblock Subharmonic projections for a quantum markov semigroup.
\newblock {\em Journal of Mathematical Physics}, February 2002.

\bibitem{FagnolaRebolledoAlgebraicCondition}
Franco Fagnola and Rolland Rebolledo.
\newblock Algebraic condition for convergence of a quantum markov semigroup to
  a steady state.
\newblock {\em Infinite dimensional analysis, Quantum probabilities and related
  topics}, 11(3):467--474, 2008.

\bibitem{FagnolaWillsUnbounded03}
Franco Fagnola and Stephen Wills.
\newblock Solving quantum stochastic differential equations with unbounded
  coefficients.
\newblock {\em Journal of Functional Analysis}, 198:279–310, 2003.

\bibitem{F.Fagnola2006}
F.Fagnola and F.Rebolledo.
\newblock {\em Open Quantum Système III, Lectures Notes in Mathematics},
  volume 2006, chapter Note on the qualitative behaviour of Quantum Markov
  Semigroups.
\newblock Springer, 2006.

\bibitem{franchiLeJan}
Jacques Franchi and Yves Le~Jan.
\newblock Relativistic diffusions and schwarzschild geometry.
\newblock 60, 02 2007.

\bibitem{Frigerio1977}
Alberto Frigerio.
\newblock Quantum dynamical semigroups and approach to equilibrium.
\newblock {\em Lett. Math. Phys.}, 2:79--87, 1977.

\bibitem{Frigerio1982}
Alberto Frigerio and Maurizio Verri.
\newblock Long time asymptotic properties of dynamical semigroups on
  w*-algebras.
\newblock {\em Mathematische Zeitschrift}, pages 275--286, 1982.

\bibitem{Gaudin1960}
Michel Gaudin.
\newblock Une démonstration simplifiée du théorème de wick en mécanique
  statistique.
\newblock {\em Nuclear Physics}, 15:89--91, 1960.

\bibitem{GoughNotes}
John Gough.
\newblock An introduction to quantum filtering.
\newblock ArXiv:1804.09086.

\bibitem{Gough12}
John~E. Gough, Matthew~R. James, and Hendra~I. Nurdin.
\newblock Single photon quantum filtering using non-markovian embeddings.
\newblock {\em Philosophical Transactions of the Royal Society of London A:
  Mathematical, Physical and Engineering Sciences}, 370(1979):5408--5421, 2012.

\bibitem{JPW14}
V.~Jak{\v{s}}i{\'{c}}, C.-A. Pillet, and M.~Westrich.
\newblock Entropic fluctuations of quantum dynamical semigroups.
\newblock {\em Journal of Statistical Physics}, 154(1):153--187, Jan 2014.

\bibitem{KarevskiPlatini2009}
Dragi Karevski and Thierry Platini.
\newblock Quantum non-equilibrium steady state induced by repeated
  interactions.
\newblock {\em Phys. Rev. Letters}, 102:207207, 2009.

\bibitem{Lancaster_Rodman95}
Peter Lancaster and Leiba Rodman.
\newblock {\em Algebraic Riccati equations}.
\newblock Oxford Science Publication, 1995.

\bibitem{Lieb1961}
Elliot Lieb, Theodore Schultz, and Daniel Mattis.
\newblock Two soluble models of an antiferromagnetic chain.
\newblock {\em Annals of Physics}, 16:407--466, 1961.

\bibitem{Lindblad1976}
G.~Lindblad.
\newblock On the generators of quantum dynamical semigroups.
\newblock {\em Commun. math. phys.}, 48:119--130, 1976.

\bibitem{thermal_machine_small_prl10}
Noah Linden, Sandu Popescu, and Paul Skrzypczyk.
\newblock How small can thermal machines be? the smallest possible
  refrigerator.
\newblock {\em Phys. Rev. Lett.}, 105:130401, Sep 2010.

\bibitem{BoutenGutaMaasen}
Madalin~Guta Luc~Bouten and Hans Maasen.
\newblock Stochastic schrödinger equations.
\newblock {\em Journal of Physics A}, 2004.

\bibitem{decoherence_assisted_transport}
Adriana Marais, Ilya Sinayskiy, Alastair Kay, Francesco Petruccione, and Artur
  Ekert.
\newblock Decoherence-assisted transport in quantum networks.
\newblock 15:013038, 01 2013.

\bibitem{Pellegrini_Nechita09}
Ion Nechita and Cl\'ement Pellegrini.
\newblock Quantum trajectories in random environment: The statistical model for
  a heat bath.
\newblock {\em Confluentes Mathematici}, 01(02):249--289, 2009.

\bibitem{Parthasarathy1992}
KR~Parthasarathy.
\newblock {\em An introduction to quantum stochastic calculus}.
\newblock Birkh\"auser-Verlag, Basel, 1992.

\bibitem{BushLahtiMittelstaedt}
P.Busch, P.J. Lahti, and P.Mittelstaedt.
\newblock {\em The Quantum Theory of Measurement}.
\newblock Springer-Verlag, 1991.

\bibitem{PellegriniDiffusive08}
Cl\'ement Pellegrini.
\newblock Existence, uniqueness and approximation of a stochastic schr\"odinger
  equation: The diffusive case.
\newblock {\em The Annals of Probability}, 36(6):2332--2353, 2008.

\bibitem{Pellegrini_2010_jumps_2levels}
Cl\'ement Pellegrini.
\newblock Existence, uniqueness and approximation of the jump-type stochastic
  schrödinger equation for two-level systems.
\newblock {\em Stochastic Processes and their Applications}, 120(9):1722 --
  1747, 2010.

\bibitem{Peschel2003}
Ingo Peschel.
\newblock Calculation of reduced density matrix from correlation functions.
\newblock {\em J. Phys. A : Math. Gen.}, 36:L205, 2003.

\bibitem{Platini2008}
Thierry Platini.
\newblock {\em \em Cha\^ines de spin quantiques hors de l'\'equilibre}.
\newblock PhD thesis, Universit\'e Henri Poincar\'e, Nancy-I, 2008.

\bibitem{Prosen2008}
Tomas Prosen.
\newblock Third quantization: a general method to solve master equations for
  quadratic open fermi systems.
\newblock {\em New Journal of Physics}, 10:043026, 2008.

\bibitem{Prosen2010}
Tomas Prosen and Bojan Zunkovic.
\newblock Exact solution of markovian master equations for quadratic fermi
  systems: thermal baths, open xy spin chains and non-equilibrium phase
  transition.
\newblock {\em New Journal of Physics}, 12:025016, 2010.

\bibitem{reebWolf_Landauer}
David Reeb and Michael~M Wolf.
\newblock An improved landauer principle with finite-size corrections.
\newblock {\em New Journal of Physics}, 16(10):103011, 2014.

\bibitem{RevuzYor}
Daniel Revuz and Marc Yor.
\newblock {\em Continuous Martingales and Brownian Motion}.
\newblock Springer.

\bibitem{RivasHuelga}
Angel Rivas and Susana Huelga.
\newblock {\em Open Quantum Systems. An Introduction}.
\newblock Springer, 2011.

\bibitem{Sachdev}
Subir Sachdev.
\newblock {\em Quantum phase transitions}.
\newblock Cambridge University Press.

\bibitem{SinayskiyPetruccioneFrancesco_control_OQBM}
Ilya Sinayskiy and Francesco Petruccione.
\newblock Steady-state control of open quantum brownian motion.
\newblock {\em Fortschritte der Physik}, 65(6-8):1600063.

\bibitem{SinayskiyPetruccione_micro_derivation}
Ilya Sinayskiy and Francesco Petruccione.
\newblock Microscopic derivation of open quantum brownian motion: a particular
  example.
\newblock {\em Physica Scripta}, 2015(T165):014017, 2015.

\bibitem{thermal_machine_small_maximal_efficiency_jpa11}
Paul Skrzypczyk, Nicolas Brunner, Noah Linden, and Sandu Popescu.
\newblock The smallest refrigerators can reach maximal efficiency.
\newblock {\em Journal of Physics A: Mathematical and Theoretical},
  44(49):492002, 2011.

\bibitem{OQS}
Alain~Joye St\'ephane~Attal and Claude-Alain Pillet, editors.
\newblock {\em Open Quantum Systems I, II, III}, volume 1880, 1881, 1882 of
  {\em Lecture Notes in Mathematics}.
\newblock Springer, 2006.

\bibitem{Stinespring55}
W.F Stinespring.
\newblock Positive functions on c*-algebras.
\newblock {\em Proceedings of the AMS}, 1955.

\bibitem{TakesakiBook}
Masami Takesaki.
\newblock {\em Theory of Operator Algebras I-II-III}.
\newblock Springer-Verlag, 2001.

\bibitem{spikes1}
Antoine Tilloy, Michel Bauer, and Denis Bernard.
\newblock Spikes in quantum trajectories.
\newblock 92, 10 2015.

\bibitem{vershik16}
A.~M. Vershik, P.B Zatiskii, and F.V Petrov.
\newblock Integration of virtually continuous functions over bistochastic
  measures and the trace formula for nuclear operators.
\newblock {\em St. Petersburg Math. J.}, 27(3):393–398, 2016.

\bibitem{dividing-channel}
Michael~M. Wolf and J.~Ignacio Cirac.
\newblock Dividing quantum channels.
\newblock {\em Communications in Mathematical Physics}, 279(1):147--168, Apr
  2008.

\end{thebibliography}

\end{document}